\documentclass[letterpaper,11pt]{article}

\usepackage[margin=1in]{geometry}  
\usepackage{xspace,exscale,relsize}
\usepackage{fancybox,shadow}
\usepackage{graphicx}
\usepackage{color}
\usepackage{amsthm,amsfonts}
\usepackage{amssymb}
\usepackage{authblk}
\usepackage[title]{appendix}

\usepackage[noadjust]{cite}

\usepackage{amsmath}
	\makeatletter
	\let\over=\@@over \let\overwithdelims=\@@overwithdelims
	\let\atop=\@@atop \let\atopwithdelims=\@@atopwithdelims
  	\let\above=\@@above \let\abovewithdelims=\@@abovewithdelims
  	\makeatother
\usepackage{fixltx2e}

\usepackage{rotating}

\usepackage{ifpdf}

\usepackage{subfigure}
\usepackage{psfrag}

\usepackage{prettyref,enumerate}

\usepackage{tikz}
\usetikzlibrary{arrows}
\tikzstyle{int}=[draw, fill=blue!20, minimum size=2em]
\tikzstyle{dot}=[circle, draw, fill=blue!20, minimum size=2em]
\tikzstyle{init} = [pin edge={to-,thin,black}]

\usepackage[
            CJKbookmarks=true,
            bookmarksnumbered=true,
            bookmarksopen=true,
            colorlinks=true,
            citecolor=red,
            linkcolor=blue,
            anchorcolor=red,
            urlcolor=blue,
            pdfauthor={Polyanskiy and Wu}
            ]{hyperref}

\usepackage[all]{xy}

\newcommand{\matx}{\ensuremath{\mathcal{X}}}

\newcommand{\mate}{\ensuremath{\mathcal{E}}}

\newcommand{\maty}{\ensuremath{\mathcal{Y}}}
\newcommand{\matn}{\ensuremath{\mathcal{N}}}

\newcommand{\matg}{\ensuremath{\mathcal{G}}}

\newcommand{\mreals}{\ensuremath{\mathbb{R}}}

\ifx\eqref\undefined
	\newcommand{\eqref}[1]{~(\ref{#1})}
\fi
\ifx\mod\undefined
	\def\mod{\mathop{\rm mod}}
\fi

\newcommand{\vect}[1]{{\bf #1}}

\def\esssup{\mathop{\rm esssup}}
\def\argmin{\mathop{\rm argmin}}

\def\exp{\mathop{\rm exp}}

\def\EE{\Expect}
\DeclareMathOperator\sign{\rm sign}

\def\Var{\mathrm{Var}}

\def\PP{\mathbb{P}}
\def\QQ{\mathbb{Q}}

\def\eqdef{\triangleq}

\def\upto{\nearrow}
\def\downto{\searrow}

\newcommand{\thetalb}{\theta_{lb}}
\newcommand{\thetac}{\theta_{c}}
\newcommand{\etaKL}{\eta_{\rm KL}}
\newcommand{\etachi}{\eta_{\chi^2}}
\newcommand{\etaTV}{\eta_{\rm TV}}

\newcommand{\FTV}{{F_{\TV}}}
\newcommand{\maxcor}{\mathrm{maxcor}}
\newcommand{\Unif}{\mathrm{Uniform}}

\newcommand{\hW}{\hat{W}}
\newcommand{\cost}{\sfM}
\newcommand{\costi}{\cost^{-1}}
\newcommand{\floor}[1]{{\left\lfloor {#1} \right \rfloor}}

\newcommand{\reals}{\mathbb{R}}
\newcommand{\naturals}{\mathbb{N}}
\newcommand{\integers}{\mathbb{Z}}
\newcommand{\Expect}{\mathbb{E}}
\newcommand{\expect}[1]{\mathbb{E}\left[#1\right]}
\newcommand{\Prob}{\mathbb{P}}
\newcommand{\prob}[1]{\mathbb{P}\left[#1\right]}
\newcommand{\pprob}[1]{\mathbb{P}[#1]}

\newcommand{\TV}{{\sf TV}}

\newcommand{\eexp}{{\rm e}}

\newcommand{\diff}{{\rm d}}
\newcommand{\eg}{e.g.\xspace}
\newcommand{\ie}{i.e.\xspace}
\newcommand{\iid}{i.i.d.\xspace}

\newcommand{\fracd}[2]{\frac{\diff #1}{\diff #2}}

\newcommand{\pth}[1]{\left( #1 \right)}
\newcommand{\qth}[1]{\left[ #1 \right]}
\newcommand{\sth}[1]{\left\{ #1 \right\}}
\newcommand{\bpth}[1]{\Bigg( #1 \Bigg)}

\newcommand{\iiddistr}{{\stackrel{\text{\iid}}{\sim}}}
\newcommand{\var}{\mathsf{var}}
\newcommand\indep{\protect\mathpalette{\protect\independenT}{\perp}}
\def\independenT#1#2{\mathrel{\rlap{$#1#2$}\mkern2mu{#1#2}}}
\newcommand{\Bern}{\text{Bern}}

\newcommand{\indc}[1]{{\mathbf{1}_{\left\{{#1}\right\}}}}
\newcommand{\Indc}{\mathbf{1}}

\definecolor{myblue}{rgb}{.8, .8, 1}
\definecolor{mathblue}{rgb}{0.2472, 0.24, 0.6} 
\definecolor{mathred}{rgb}{0.6, 0.24, 0.442893}
\definecolor{mathyellow}{rgb}{0.6, 0.547014, 0.24}

\newcommand{\tX}{{\tilde{X}}}

\newcommand{\sfM}{{\mathsf{M}}}

\newcommand{\sfQ}{{\mathsf{Q}}}

\newcommand{\calF}{{\mathcal{F}}}
\newcommand{\calG}{{\mathcal{G}}}

\newcommand{\calN}{{\mathcal{N}}}

\newcommand{\calW}{{\mathcal{W}}}
\newcommand{\calX}{{\mathcal{X}}}

\newcommand{\hX}{\hat{X}}

\newcommand{\Th}{{^{\rm th}}}
\newcommand{\diverge}{\to \infty}

\newrefformat{eq}{(\ref{#1})}
\newrefformat{thm}{Theorem~\ref{#1}}
\newrefformat{th}{Theorem~\ref{#1}}
\newrefformat{chap}{Chapter~\ref{#1}}
\newrefformat{sec}{Section~\ref{#1}}
\newrefformat{algo}{Algorithm~\ref{#1}}
\newrefformat{fig}{Fig.~\ref{#1}}
\newrefformat{tab}{Table~\ref{#1}}
\newrefformat{rmk}{Remark~\ref{#1}}
\newrefformat{clm}{Claim~\ref{#1}}
\newrefformat{def}{Definition~\ref{#1}}
\newrefformat{cor}{Corollary~\ref{#1}}
\newrefformat{lmm}{Lemma~\ref{#1}}
\newrefformat{prop}{Proposition~\ref{#1}}
\newrefformat{pr}{Proposition~\ref{#1}}
\newrefformat{app}{Appendix~\ref{#1}}
\newrefformat{apx}{Appendix~\ref{#1}}
\newrefformat{ex}{Example~\ref{#1}}
\newrefformat{exer}{Exercise~\ref{#1}}
\newrefformat{soln}{Solution~\ref{#1}}

\def\unifto{\mathop{{\mskip 3mu plus 2mu minus 1mu%
	\setbox0=\hbox{$\mathchar"3221$}%
	\raise.6ex\copy0\kern-\wd0%
	\lower0.5ex\hbox{$\mathchar"3221$}}\mskip 3mu plus 2mu minus 1mu}}

\ifx\lesssim\undefined
\def\simleq{{{\mskip 3mu plus 2mu minus 1mu%
	\setbox0=\hbox{$\mathchar"013C$}%
	\raise.2ex\copy0\kern-\wd0%
	\lower0.9ex\hbox{$\mathchar"0218$}}\mskip 3mu plus 2mu minus 1mu}}
\else
\def\simleq{\lesssim}
\fi

\ifx\gtrsim\undefined
\def\simgeq{{{\mskip 3mu plus 2mu minus 1mu%
	\setbox0=\hbox{$\mathchar"013E$}%
	\raise.2ex\copy0\kern-\wd0%
	\lower0.9ex\hbox{$\mathchar"0218$}}\mskip 3mu plus 2mu minus 1mu}}
\else
\def\simgeq{\gtrsim}
\fi

\def\dperp{\perp\!\!\!\perp}

\newtheorem{theorem}{Theorem}
\newtheorem{lemma}[theorem]{Lemma}
\newtheorem{corollary}[theorem]{Corollary}

\newtheorem{proposition}[theorem]{Proposition}
\newtheorem{prop}[theorem]{Proposition}

\theoremstyle{definition}

\newtheorem{remark}{Remark}

\newif\ifmapx
{\catcode`/=0 \catcode`\\=12/gdef/mkillslash\#1{#1}}
\edef\jobnametmp{\expandafter\string\csname contraction_apx\endcsname}
\edef\jobnameapx{\expandafter\mkillslash\jobnametmp}
\edef\jobnameexpand{\jobname}
\ifx\jobnameexpand\jobnameapx
\mapxtrue
\else
\mapxfalse
\fi

\long\def\apxonly#1{\ifmapx{\color{blue}#1}\fi}

\begin{document}
\ifpdf
\DeclareGraphicsExtensions{.pgf}
\graphicspath{{figures/}{plots/}}
\fi

\title{Dissipation of information in channels with input constraints}

\author{Yury Polyanskiy and Yihong Wu\thanks{Y.P. is with the Department of EECS, MIT, Cambridge, MA, \url{yp@mit.edu}. Y.W. is with
the Department of ECE, University of Illinois at Urbana-Champaign, Urbana, IL, \url{yihongwu@illinois.edu}. The research of Y.P. has been supported by the Center for Science of Information (CSoI),
an NSF Science and Technology Center, under grant agreement CCF-09-39370 and by the NSF CAREER award under grant
agreement CCF-12-53205. 
The research of Y.W. has been supported in part by NSF grants IIS-1447879 and CCF-1423088.
}}


\maketitle

\begin{abstract}
One of the basic tenets in information theory, the data processing inequality states that output divergence does not
exceed the input divergence for any channel.  For channels without input constraints, various estimates on the amount of
such contraction are known, Dobrushin's coefficient for the total variation being perhaps the most well-known. This work
investigates channels with average input cost constraint. It is found that while the contraction coefficient typically
equals one (no contraction), the information nevertheless dissipates. A certain non-linear function, the \emph{Dobrushin
curve} of the channel, is proposed to quantify the amount of dissipation. Tools for evaluating the Dobrushin curve of
additive-noise channels are developed based on coupling arguments. Some basic applications in 
stochastic control, uniqueness of Gibbs measures and fundamental limits of noisy circuits are discussed.

As an application, it shown that in the chain of $n$ power-constrained relays and Gaussian channels the end-to-end mutual information and maximal squared correlation decay as $\Theta(\frac{\log\log n}{\log n})$, which is in stark contrast with the exponential decay in chains of discrete channels. Similarly, the behavior of noisy circuits (composed of gates with bounded fan-in) and broadcasting of information on trees (of bounded degree)
does not experience threshold behavior in the signal-to-noise ratio (SNR). Namely, unlike the case of discrete channels, the
probability of bit error stays bounded away from $1\over 2$ regardless of the SNR.
\end{abstract}

\tableofcontents


%

\section{Introduction}

Consider the following Markov chain
\begin{equation}\label{eq:mc1}
		W \to X_1 \to Y_1 \to X_2 \to Y_2 \to \cdots \to X_n \to Y_n\,,
\end{equation}
where the random variable $W$ is the original message (which is to be estimated on the basis of $Y_n$ only), each $P_{Y_j|X_j}$ is a
standard vector-Gaussian channel of dimension $d$:
\begin{equation}\label{eq:mc2}
		P_{Y_j|X_j=\vect x} = \matn(\vect x, \vect I_d) 
\end{equation}	
	and each input $X_j$ satisfies a power constraint:
\begin{equation}\label{eq:mc3}
			\EE[\|X_j\|^2] \le d E\,.
\end{equation}	
The goal is to design the transition kernels $P_{X_{j+1}|Y_j}$, which we refer to as processors or encoders, to
facilitate the estimation of $W$ at the end of the chain. See Fig.~\ref{fig:control} for an illustration.

Intuitively, at each stage some information about the original message $W$ is lost due to the external noise. Furthermore, each processor cannot de-noise completely due to the finite power constraint. Therefore it is reasonable to expect that for very large $n$ we should have
	$$ P_{W, Y_n} \approx P_W P_{Y_n}, $$
that is, $W$ and $Y_n$ become almost independent. We quantify this intuition in terms of the total variation,
Kullback-Leibler (KL) divergence and correlation, namely
\begin{align} \TV(P, Q) &\eqdef \sup_E |P[E] - Q[E]| = {1\over 2} \int |\diff P - \diff Q|, \label{eq:tv}\\
	   D(P\|Q) &\eqdef \int \log {\diff P\over \diff Q}\, \diff P,\\
	   \rho(A, B) &\eqdef {\EE[AB] - \EE[A] \EE[B]\over \sqrt{ \Var[A] \Var[B]}},\\
	   I(A;B) &\eqdef D(P_{A,B} \| P_A P_{B}).
   \end{align}
Our main result is the following theorem, which shows that the information about the original message is eventually lost in both an information-theoretic and an estimation-theoretic sense.
\begin{theorem}\label{th:main}
	Let $W, X_j, Y_j$ for a Markov chain as in~\eqref{eq:mc1} -- \eqref{eq:mc3}. Then
	\begin{align} \TV(P_{W Y_n}, P_W P_{Y_n}) \leq \frac{CdE}{\log n}  &\to 0,\label{eq:main1} \\
		I(W; Y_n) \leq C'd^2 E\cdot \frac{\log \log n}{\log n} &\to 0, \label{eq:main3}\\
		\sup_{g\in L_2(P_{Y_n})} \rho(W, g(Y_n))  &\to 0,   \label{eq:main2}
   \end{align}	   
   where $C,C'>0$ are some universal constants.   
   Moreover, the right-hand side of~\eqref{eq:main2} is $O(\frac{1}{\sqrt{\log n}})$ if $W$ is finitely valued and 
   $O(\sqrt{\frac{\log \log   n}{\log n}})$ if $W$ is sub-Gaussian, respectively.
\end{theorem}

When $W$ is scalar Gaussian, all estimates of the convergence rates in \prettyref{th:main} are sharp, in the sense that there exists a sequence of power-constrained relay functions such that $\TV(P_{W Y_n}, P_W P_{Y_n})) = \Omega(\frac{1}{\log n})$, $I(W; Y_n) = \Omega(\frac{\log \log n}{\log n})$ and $\sup_{g\in L_2(P_{Y_n})} \rho(W, g(Y_n)) = \Omega(\sqrt{\frac{\log\log n}{\log n}})$.

Our interest in the problem has been mainly motivated by the fact that the moment constraint~\eqref{eq:mc3} renders the
standard tools for estimating convergence rates of information measures inapplicable. Thus a few new ideas are developed
in this paper.  In order to explain this subtlety, it is perhaps easiest to contrast Theorem~\ref{th:main} and
especially~\eqref{eq:main3} with the recent results of Subramanian et al.~\cite{Subramanian12,SVL13} on cascades of AWGN
channels. Other applications of our techniques are deferred till Section~\ref{sec:appl}.

In~\cite{Subramanian12,SVL13} an upper estimate on $I(W; Y_n)$ is derived under extra constraints on relay functions.
Among these constraints, the most important one is that the average constraint~\eqref{eq:mc3} is
replaced with a seemingly similar one:
\begin{equation}\label{eq:mc3sub}
	\|X_j\|^2 \le d E \qquad \text{a.s.}
\end{equation}
It turns out, however, that for the analysis of~\eqref{eq:mc3sub} the standard tools (in particular the Dobrushin
contraction coefficient) not only recover all the results of~\cite{Subramanian12,SVL13} but in fact simplify and
strengthen them. Thus, we start with describing those classical methods in the next section, and describe how to
analyze~\eqref{eq:mc3sub} in Section~\ref{sec:subram} to follow. 


\textit{Added in print:} A completely different method (without recoursing to the total variation) for showing~\eqref{eq:main3} has been developed in~\cite{yuryITA,CPW15} based on strong data processing inequalities for mutual information in Gaussian noise.

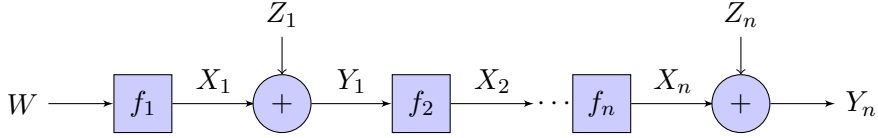
\begin{figure}[t]
	\centering
	\ifpdf
\begin{tikzpicture}[scale=1,transform shape,node distance=1.8cm,auto,>=latex']
    \node [int] (f1) {$f_1$}; 
    \node  (Z0) [left of=f1, node distance=1.6cm] {$W$};
    \node (X0) [left of=Z0,node distance=1.5cm, coordinate] {};
    \node [coordinate] (X1) [right of=f1, node distance=1.5cm]{};
    \path[->] (Z0) edge node {} (f1);
    \draw[->] (f1) edge node {$X_1$} (X1);
    
    \node [dot, pin={[init]above:{$Z_1$}}] (Z1) [right of=f1,node distance=1.85cm] {$+$};
    \node [int] (f2) [right of=Z1,node distance=1.85cm] {$f_2$};    
    \path[->] (Z1) edge node {$Y_1$} (f2);
    \node [coordinate] (X2) [right of=f2, node distance=1.5cm]{};
    \draw[->] (f2) edge node {$X_2$} (X2);
    
    \node (dots) [right of=X2,node distance=0.3cm] {$\cdots$};
    \node [int] (fn) [right of=dots,node distance=0.6cm] {$f_n$};  
    \node [coordinate] (Xn) [right of=fn, node distance=1.5cm]{};
    \draw[->] (fn) edge node {$X_n$} (Xn);
    \node [dot, pin={[init]above:{$Z_n$}}] (Zn) [right of=fn,node distance=1.85cm] {$+$};
    \node (Yn) [right of=Zn,node distance=1.6cm] {$Y_n$};    
    \path[->] (Zn) edge node {} (Yn);
\end{tikzpicture}
	\else
	\vbox to 40pt{\hsize=.8\columnwidth\Huge Tikz diagram}
	\fi
	\caption{Cascade of AWGN channels with power-constrained relays $f_j$.}
	\label{fig:control}
\end{figure}

\subsection{Contraction properties of Markov kernels}
Fix a transition probability kernel (channel) $P_{Y|X}: \matx\to\maty$ acting between two measurable spaces. We denote
by $P_{Y|X}\circ P$ the distribution on $\maty$ induced by the push-forward of the distribution $P$, which is the distribution of the output $Y$ when the input $X$ is distributed according to $P$,
 and by $P \times P_{Y|X}$ the joint distribution $P_{XY}$ if $P_X=P$.
We also denote by $P_{Z|Y} \circ P_{Y|X}$ the serial composition of channels. Let $f:\mreals^+\to\mreals$
be a convex function with $f(1) = 0$ and let $D_f(P||Q) \triangleq \Expect_Q[f(\fracd{P}{Q})]$ denote the corresponding $f$-divergence,
cf.~\cite{IC67}. For example taking $f(x)=(x-1)^2$ we obtain the $\chi^2$-divergence:
\begin{equation} \chi^2(P\|Q) \eqdef \int \left({\diff P\over \diff Q}\right)^2 \diff Q - 1\label{eq:chi2def}\,. 
\end{equation}
For any $Q$ that is not a point mass, define:
\begin{align}
\eta_f(Q) &\eqdef  \sup_{P: 0<D_f(P\|Q)<\infty} {D_f(P_{Y|X} \circ P \| P_{Y|X} \circ Q)\over
			D_f(P\| Q)}, \label{eq:eta_fq}\\
			\eta_f &\eqdef \sup_Q \eta_f(Q)\,.\label{eq:eta_f} 
   \end{align}
   For $f(x)=|x-1|$, $f(x)=(x-1)^2$ and $f(x)=x \log x$ we will write
   $\etaTV(Q), \etachi(Q)$ and $\etaKL(Q)$, respectively. 
   In particular, $\etaTV$ is known as the \emph{Dobrushin's coefficient} of the kernel $P_{Y|X}$, 
   which is one of the main tools for studying ergodicity property of Markov chains as well as Gibbs measures.

   \paragraph{General alphabets} Dobrushin~\cite{RLD56} showed that supremum in the definition of $\etaTV$ can be
   restricted to single-point distributions
   $P$ and $Q$, thus providing a simple criterion for strong ergodicity of Markov processes. 
   It is
   well-known, e.g.~Sarmanov \cite{OS58}, that $\eta_{\chi^2}(Q)$ is the squared \emph{maximal correlation}
   coefficient of the joint distribution $P_{XY}=Q\times P_{Y|X}$:
   \begin{equation}\label{eq:chi_maxcor0}
	   S(X;Y) \eqdef \sup_{f, g} \rho(f(X), g(Y)) = \sqrt{\eta_{\chi^2}(Q)}\,.
\end{equation}	
   Later~\cite[Proposition II.4.10]{CKZ98} (see also \cite[Theorem 4.1]{CIR93} for finite alphabets)
   demonstrated that all other contraction coefficients
   are upper-bounded by the Dobrushin's coefficient $\etaTV$:
   \begin{align}\label{eq:eta_ub}
	   	\eta_f \le \etaTV,
\end{align}	
and this inequality is typically strict.\footnote{E.g. for the binary symmetric channel with crossover probability
$\delta$ we have $\eta_{\chi^2}=\etaKL=(1-2\delta)^2<\etaTV=|1-2\delta|$.}
In the opposite direction it can be shown, cf.~\cite[Proposition II.6.15]{CKZ98},
	\begin{equation}\label{eq:eta_lb}
		\eta_{\chi^2} \le \eta_f\,,
\end{equation} 
whenever $f$ is thrice differentiable with $f''(1) > 0$. Moreover, \prettyref{eq:eta_lb} holds with equality for all
nonlinear and operator convex $f$, \eg, for KL divergence and for squared Hellinger distance; see~\cite[Theorem
1]{CRS94} and \cite[Proposition II.6.13 and Corollary II.6.16]{CKZ98}. In particular,
   \begin{equation}\label{eq:chi2KL}
	   	\eta_{\chi^2} = \etaKL, 
   \end{equation}   	
   which was first obtained in \cite{AG76} using different methods.
	\apxonly{More generally, the range of 
	$f\mapsto \eta_f$ is included (\textit{and equal} for discrete channels) in
	$$ \sup_{x \neq x'} \beta_1(P_x, P_{x'}) \le \eta_f \le \etaTV $$
	see~\cite[Propositions II.6.3 and II.6.4]{CKZ98}.}
	Rather naturally, we also have~\cite[Proposition II.4.12]{CKZ98}:
	$$ \eta_f = 1 \quad \iff \quad \etaTV = 1 $$
	for any non-linear $f$.

The fixed-input contraction coefficient $\etaKL(Q)$ is closely related to the (modified) log-Sobolev inequalities. Indeed, 
when $Q$ is invariant under $P_{Y|X}$ (i.e. $P_{Y|X}\circ Q=Q$) any initial distribution $P$
	converges to $Q$ exponentially fast in terms of $D(P_{Y|X}^n \circ
   P || Q)$ with exponent upper-bounded by $\etaKL(Q)$, which in turn can be
   estimated from log-Sobolev inequalities, e.g.~\cite{MLxx}. When $Q$ is not invariant,
   it was shown~\cite{MLM03} that
   \begin{equation}\label{eq:disc_miclo}
	   	1-\alpha(Q) \le \etaKL(Q) \le 1-C\alpha(Q)
   \end{equation}   
   holds for some universal constant $C$, where $\alpha(Q)$ is a modified log-Sobolev (also known as $1$-log-Sobolev) constant:%
\apxonly{\footnote{
When $Q$ is invariant measure, and in continuous time (i.e. for a Markov process) we in fact have
$\etaKL(Q)=1-\alpha(Q)$, where $\etaKL$ this time is the smallest constant satisfying
$$ \left.{d\over dt}\right|_{t=0} D(\pi_t || Q)\le -(1-\etaKL(Q)) D(\pi_0 || Q)\,.$$
   Indeed, if $P_t=e^{-tL}$ is a semigroup, and $\mate(f, g) = \EE[f(X) Lg(X)]$ its
   Dirichlet form, then 
   \begin{equation}\label{eq:dif_ent}
	   	\left.{d\over dt}\right|_{t=0} D(\pi_t || Q) = -\mate\left({d\pi_t\over \diff Q}, \log {dP_t Q\over
   \diff Q}\right) 
   \end{equation}   
   and the smallest ratio $\mate(f, \log f)\over \EE_Q[f\log f]$ is precisely the
   1-log-Sobolev constant.~\eqref{eq:disc_miclo} extends this in two ways: first a usual
   Miclo-extension to discrete time is made, cf.~\cite{LM97}, and then one also needs to
   find ${d\over dt} D(\pi_t ||Q_t)$ since $Q_t = P_t Q_0$ also changes now.}}%
   $$ \alpha(Q) = \inf_{f \perp 1, \|f\|_2=1} {\EE\left[ f^2(X) \log {f^2(X)\over
   f^2(X')}\right] \over \EE[f^2(X) \log f^2(X)] }, \qquad P_{X X'} = Q \times
   (P_{X|Y} \circ P_{Y|X}) .$$

	\paragraph{Finite alphabets}  
   Ahlswede and G\'acs~\cite{AG76} have shown
   $$ \eta_{\chi^2}(Q) < 1 \iff \etaKL(Q) < 1 \iff \text{graph~} \{(x,y): Q(x)>0,
   P_{Y|X}(y|x) > 0\}\mbox{~is connected}. $$
   As a criterion for $\eta_f(Q)<1$, this is an improvement of~\eqref{eq:eta_ub} only for channels with
   $\etaTV=1$. Furthermore,~\cite{AG76} shows
   \begin{equation}\label{eq:chi2KL_bound}
	   \eta_{\chi^2}(Q) \le \etaKL(Q), 
   \end{equation}
   with inequality frequently being strict.\footnote{See~\cite[Theorem 9]{AG76} and 
   \cite{AGKN13} for examples.} We note that the main result of~\cite{AG76} characterizes $\etaKL(Q)$
   as the maximal ratio of hyper-contractivity of the conditional expectation operator
   $\EE[\cdot|X]$. 
   For finite alphabets, \prettyref{eq:eta_lb} can be strengthened to the following fixed-input version under the same conditions on $f$ (c.f.~\cite[Theorem 3.3]{Raginsky14}):
	\begin{equation}\label{eq:eta_lb-Q}
		\eta_{\chi^2}(Q) \le \eta_f(Q)\,.
\end{equation}
   For connections between $\etaKL$ and log-Sobolev inequalities on finite alphabets see~\cite{Raginsky13}.

   \apxonly{
   \textbf{TODO}: Example for $\eta_f(Q) > \etaTV(Q)$. ????
   }

\subsection{Exponential decay of information when $\etaTV<1$}\label{sec:subram}

First, it can be shown that (See \prettyref{app:etaKL} for a proof in the general case. The finite alphabet case has been shown in \cite{AGKN13})
\begin{equation}
	\sup {I(U; Y)\over I(U;X)} = \etaKL(P_X)\,,
	\label{eq:etaKL}
\end{equation}
where the supremum is taken over all Markov chains $U\to X \to Y$ with fixed $P_{XY}$ such that $0<I(U;X)<\infty$. Thus, for an arbitrary Markov chain 
$$ W\to X_1 \to Y_1 \to X_2 \to Y_2 \to \cdots \to Y_n $$
with
equal channels $P_{Y_j|X_j} = P_{Y|X}$ for all $j$, 
we have 
\begin{equation}\label{eq:expo1}
	I(W; Y_n) \le \prod_{j=1}^n \etaKL(P_{X_j}) \cdot I(W; X_1) \le (\etaKL)^n \cdot H(W)\,.
\end{equation}
A similar argument leads to
\begin{align} \TV(P_{W Y_n}, P_W P_{Y_n}) &\le \prod_{j=1}^n \etaTV(P_{X_j}) \le (\etaTV)^n ,\\
	\rho^2(W; Y_n) \le S(W; Y_n) &\le \prod_{j=1}^n S(X_j; Y_j) \le (\eta_{\chi^2})^{n}\,. \label{eq:rho_contr_simple}
\end{align}   
Thus, in the simple case when $\etaTV<1$ we have from~\eqref{eq:eta_ub} that when $n\diverge$, all three information quantities
converge to zero exponentially as fast as $\etaTV^n$. 

Let us now consider the case of~\cite{Subramanian12,SVL13}, namely the AWGN channel $P_{Y|X}$ with maximal power
constraint~\eqref{eq:mc3sub}. First recall that
\begin{equation}\label{eq:tvnorm}
		\TV(\matn(\mu_1, \vect I_d), \matn(\mu_2, \vect I_d)) = 1-2\sfQ(|\mu_1 - \mu_2|/2)\,,
\end{equation}
where $\sfQ(x) = \int_x^\infty {1\over \sqrt{2\pi}} e^{-t^2/2} dt$ is the Gaussian complimentary CDF and $|\cdot|$ denotes the Euclidean norm. 
Then by Dobrushin's characterization of $\etaTV$ we get that for any $P_{X_j}$ satisfying~\eqref{eq:mc3sub} we have
$$ \etaTV = \sup_{x_1, x_2} 1-2\sfQ(|x_1 - x_2|/2) = 1-2\sfQ(\sqrt{d E})\,.$$
From~\eqref{eq:expo1} this implies
\begin{equation}\label{eq:fuck_sub}
	I(W; Y_n) \le (1-2\sfQ(\sqrt{d E}))^n \cdot H(W)\,. 
\end{equation}

It turns out~\eqref{eq:fuck_sub} is stronger than the main result of~\cite{SVL13} and independent of the cardinality of
$W$. Indeed, although~\cite{SVL13} did not point
this out, the analysis there corresponds to the following upper-bound on $\etaTV$
\begin{equation}\label{eq:tvbd}
		\etaTV \le 1 - \sum_{y \in \maty} \inf_{x\in\matx} P_{Y|X}(y|x) 
\end{equation}
(here we assumed finite alphabet $\maty$ for simplicity). This bound is clearly tight for the case of $|\matx| = 2$ but
 rather loose for larger $|\matx|$. Since we calculated $\etaTV$ exactly,~\eqref{eq:fuck_sub} must yield a better bound
 than that of~\cite{SVL13}. However, the estimate \prettyref{eq:fuck_sub} relies on the Dobrushin coefficient, which, as will be shown below, breaks down if the power constraints is imposed on average instead of almost surely. To remedy this problem requires developing new tools to complement the Dobrushin coefficient. For the generalization to average power constraint as well as discussions for multi-hop communication, see \prettyref{prop:iconv} and \prettyref{rmk:rate-pos} in \prettyref{sec:I}.

\apxonly{I also checked that exponential approximation given in~\cite[(18)]{SVL13} is indeed worse for all $R$ then our
$$ I(W;Y_n) \leq (1-e^{-d E/2 + o(d)})^n\,. $$
The way to prove~\eqref{eq:tvbd} is by noticing that RHS=LHS for $|\matx|=2$. Thus we first include those
$x_1,x_2$ that attain $\etaTV$ and then start adding other symbols -- LHS does not change, while RHS grows.
}%

The main part of this paper handles convergence of $I(W; Y_n)\to0$ in the case~\eqref{eq:mc3}, for which unfortunately
$\etaTV=\etaKL=\eta_{\chi^2}=1$. Indeed, by taking
\begin{align} P &= (1-t) \delta_0 + t \delta_a\,,\\
	Q &= (1-t) \delta_0 + t \delta_{-a}\,,
\end{align}
and performing a straightforward calculation, we find
\begin{equation}\label{eq:rt1}
	{\TV(P*\matn(0,1), Q*\matn(0,1))\over \TV(P,Q)} = 1-2\sfQ(a) \xrightarrow{a\diverge} 1.
\end{equation}
Therefore, even if one restricts the supremum in~\eqref{eq:eta_f} to $P$ and $Q$ satisfying the moment
constraint~\eqref{eq:mc3} (in fact, any constraint on the tails for that matter), choosing
$a \diverge$ and $t\to 0$ accordingly drives the ratio in \prettyref{eq:rt1} to one, thus proving $\etaTV=1$. 
This example is instructive: The ratio~\eqref{eq:rt1} approaches $1$ only when the $\TV(P,Q)\to0$. Our idea is to 
get \emph{non-multiplicative} contraction inequalities that still guarantee strict decrease of total variation after convolution.

Similarly, there is no moment condition which can guarantee the strict contraction of the KL divergence or mutual information. For example, 
it can be shown that
$$	\sup {I(U; X+Z)\over I(U;X)} = 1\,,$$
where the supremum is over all Markov chains $U\to X \to X+Z$ with $\Expect[|X|^2] \leq 1$.
This suggests that the exponential decay of mutual information in \prettyref{eq:expo1} obtained under peak power constraint might fail.
Indeed, we will show that under average power constraint, the decay speed of mutual information can be much slower than exponential (see \prettyref{sec:ach}).

\subsection{Organization}
The rest of the paper is organized as follows.
Section~\ref{sec:curve} proves results on reduction of total variation over additive-noise channels; we call the
resulting relation the \textit{Dobrushin curve} of a channel. Section~\ref{sec:mate} shows how to convert knowledge about
total variation to other $f$-divergences, extending~\eqref{eq:eta_ub}.
Section~\ref{sec:proof} shows how to use Dobrushin curve to prove Theorem~\ref{th:main}. Finally, Section~\ref{sec:appl}
concludes with applications (other than Theorem~\ref{th:main}).

In particular, in \prettyref{sec:control} we show that the optimal correlation achieved by non-linear control in the $n$-stage Gaussian quadratic control problem studied by Lipsa and Martins \cite{LM11} is 
$\Theta(\sqrt{\frac{\log \log n}{\log n}})$; in contrast, the best linear controller only achieves exponentially small correlation. 
The inferiority of linear control can be explained from the viewpoint of dissipation of information and contraction of KL divergence.
In \prettyref{sec:gibbs} we extend Dobrushin's strategy for proving uniqueness of Gibbs measures to unbounded systems with
moment constraints on marginal distributions. And in \prettyref{sec:circuits} we apply our technique to proving a lower
bound on the probability of error in circuits of noisy gates. 

Finally, in \prettyref{sec:trees} we show that in the 
question of broadcasting a single bit on a tree of 
Gaussian channels there is no phase transition. Namely, for arbitrarily low SNR it is possible to build relays
satisfying the average power constraint so that given the received values on all leaves at depth $d$ 
the probability of error of estimating the original bit is bounded away from $1/2$. This is in contrast to the case of
trees of binary symmetric channels, studied by Evans-Kenyon-Peres-Schulman~\cite{EKPS00}, who showed that there there
is a phase transition in terms of the strength of the channel noise.

\section{Dobrushin curve of additive-noise channels} \label{sec:curve}
\subsection{Definitions and examples}

Let $P_{Y|X}:\matx \to \maty$ be a probability transition kernel. 
Then, we define the \emph{Dobrushin curve} of $P_{Y|X}$ as follows: 
\begin{equation}	
	\FTV(t) = \sup\{ \TV(P_{Y|X} \circ P, P_{Y|X} \circ Q): \TV(P, Q) \le t, (P,Q) \in  \calG\},  \qquad t\in[0,1]
	\label{eq:Ft}
\end{equation}
where $\calG$ is some (convex) set of pairs of probability measures. The curve $t \mapsto \FTV(t)$ defines the upper boundary of the region
\begin{equation}
\calF_{\TV} = \left\{\big(\TV(P_{Y|X} \circ P, P_{Y|X} \circ Q), \TV(P, Q)\big): (P,Q) \in  \calG\right\}
	\subset	[0,1]^2,
	\label{eq:calF}
\end{equation}
which is the joint range of the input and output total variations.

We notice the following ``data-processing'' property of Dobrushin curves: if $\FTV_1$  and $\FTV_2$ are the Dobrushin
curves of channels $P_{Y_1|X_1}$ and $P_{Y_2|X_2}$ (and the respective feasible sets $\calG_1$ and $\calG_2$), then for any $P_{X_2|Y_1}$
that connects them:
$$ X_1 \stackrel{P_{Y_1|X_1}}{\longrightarrow} Y_1 \longrightarrow X_2 \stackrel{P_{Y_2|X_2}}{\longrightarrow}
Y_2 $$
we naturally have for the combined channel
$$ \FTV(t) \le \FTV_2(\FTV_1(t))$$
(the constraint set $\calG$ corresponding to $\FTV(t)$ is defined so that $(P,Q)\in \calG_1$ and
$(P_{X_2|Y_1}\circ P_{Y_1|X_1} \circ P, P_{X_2|Y_1}\circ P_{Y_1|X_1} \circ Q)\in\calG_2$). 
This observation will be central for the analysis of the Markov chain~\eqref{eq:mc1}. We proceed to 
computing $\FTV$.

For simplicity, in the sequel we focus our presentation on the following:
\begin{enumerate}
	\item Consider $\matx=\maty=\mreals^d$ with Borel $\sigma$-algebra and $d\in\naturals\cup\{+\infty\}$.
	\item There is a norm $|\cdot|$ on $\mreals^d$.
	\item The constraint set $\calG$ is defined by some average cost constraint:
\begin{equation}
\calG_a \triangleq \{(P,Q): \Expect_P [\cost(|X|)] + \Expect_Q[\cost(|X|)] \leq 2 a \},
	\label{eq:Ga}
\end{equation}
where $\cost: \reals_+ \to \reals_+$ is a strictly increasing convex cost function\footnote{Our motivating examples are $\cost(x) = x^p$ with $p\geq 1$, $\cost(x) = \exp(\alpha x)-1$ and
$\cost(x)=\exp(\alpha x^2)-1$ with $\alpha>0$, which we call $p\Th$-moment, sub-exponential and sub-Gaussian
constraints, respectively.} with
$\cost(0)=0$ and $a\geq 0$.
	\item The random transformation $P_{Y|X}$ acts by convolution (on $\mreals^d$) with noise $P_Z$:
		$$	P_{Y|X=x} = P_{Z+x}\qquad x,Y,Z \in \mreals^d .$$ 
\end{enumerate}


\begin{remark}\label{rmk:FTVconcave}
For any point $(\TV(P,Q),\TV(P*P_Z,Q*P_Z))$ in the region $\calF_{\TV}$ and $\lambda \in [0,1]$, we can achieve the point  $(\lambda
\TV(P,Q), \lambda \TV(P*P_Z,Q*P_Z))$ by setting $P_{\lambda} = \lambda P + (1-\lambda) \delta_0$ and $Q_{\lambda} =
\lambda Q + (1-\lambda) \delta_0$. This implies that $t \mapsto {\FTV(t)\over t}$ is non-increasing. However, 
this does not imply that $\calF_{\TV}$ is convex or that $\FTV$ is concave. Shortly, we will demonstrate that for many
noise distribution $P_Z$ the Dobrushin curve $\FTV$ is in fact concave. 
\end{remark}

Expanding on the previous remark, we can further show relations between $\FTV$ computed for different cost values of $a$
in~\eqref{eq:Ga}.
\begin{proposition}\label{prop:yw} 
	Let $\FTV(t,a)$ be the Dobrushin curve for some channel $P_{Y|X}$ and constraint~\eqref{eq:Ga}, where
	$\cost(0)=0$. Then for all $\alpha\ge0$ such that $\alpha t\le 1$ we have
	\begin{equation}\label{eq:yw0}
			\FTV(\alpha t, \alpha a) = \alpha \FTV(t, a)\,.
	\end{equation}	
	In particular, $\FTV(t,a)=t g(a/t)$, where $g(a)\eqdef \FTV(1, a)$ and in the unconstrained case Dobrushin curve
	is a straight line: $\FTV(t, \infty)=\etaTV t$.
\end{proposition}
\begin{proof} Without loss of generality, we may assume $\alpha \le 1$ (otherwise, apply to $t'=\alpha t$ and
	$\alpha'=1/\alpha$). For all $s\in[0,1]$ we have two inequalities
	\begin{align} \FTV(st, sa) &\ge s \FTV(t,a)\label{eq:yw1},\\
		\FTV(s, sa/t) &\ge s/t \FTV(t,a) \label{eq:yw2}.
	\end{align}
	To show the first start with arbitrary $(P,Q)\in\calG_a$ such that $\TV(P,Q)=t$ and $\TV(P_{Y|X}\circ P,
	P_{Y|X}\circ Q)=f$. Then we can construct distributions
	$$ P_1 = sP + (1-s) \delta_0\,,\quad  Q_1 = sQ + (1-s)\delta_0\,,$$
	for which 
	$$ \int \cost(x) (dP_1 + dQ_1) \le s a,\quad \TV(P_1, Q_1) = st,\quad \TV(P_{Y|X}\circ P_1, P_{Y|X} \circ Q_1) = sf $$
	and thus~\eqref{eq:yw1} follows after optimizing over $(P,Q)$. The second inequality follows by
	considering\footnote{Measures $(P-Q)^+$ and $(P-Q)^-$ denote the two pieces of Jordan decomposition of measure
	$(P-Q)$.}
	$$ P_2 = {s\over t}(P-Q)^+ + (1-s) \delta_0, \quad Q_2 = {s\over t} (P-Q)^- + (1-s)\delta_0 $$
	and a similar argument.
	Finally,~\eqref{eq:yw0} follows from~\eqref{eq:yw1} (with $s=\alpha$) and~\eqref{eq:yw2} (with $s=t/\alpha$).
\end{proof}

\apxonly{\textbf{Note:} If one could show that $g(a) \triangleq \sup\{\TV(P*P_Z,Q*P_Z): P\perp Q, (P,Q)\in\calG_a\}$ is
concave in $a$, then by the Proposition $\FTV$ is always concave. But this is not clear, since $P_i\perp Q_i,i=1,2 \not \Rightarrow (P_1+P_2)/2 \perp (Q_1+Q_2)/2$.}

\subsection{Criterion for $\FTV(t)<t$}

Similar to how Dobrushin's results~\cite{RLD56} reduce the computation of $\etaTV$ to considering the two-point quantity $\TV(P_{Y|X=x},
P_{Y|X=x'})$, our main tool will be the following function $\theta: \reals^d \to [0,1]$ defined by
\begin{equation}
	\theta(x) \eqdef \TV(P_Z, P_{Z+x})\,, \qquad x\in\mreals^d.
	\label{eq:theta}
\end{equation} 
Some simple properties of $\theta$ (general case) are as follows:
\begin{itemize}
	\item $\theta(0)=0$, $\lim_{x\to\infty}\theta(x)=1$.
	\item $\theta(x)=\theta(-x)$.
	\item If $P_Z$ is compactly supported then $\theta(x)=1$ when $|x|$ is sufficiently large.
	\item $\theta$ is lower-semicontinuous (since total variation is weakly
			lower-semicontinuous).
	\item If $P_Z$ has a density $f_Z$, then
	$$ \theta(x) = \int_{\mreals^d} |f_Z(z-x) -f_Z(z)| \diff z. $$
		and $\theta$ is continuous on $\mreals$, which follows from the denseness of compactly-supported continuous functions 
		in $L_1(\reals^d)$. 
		\apxonly{Proof: approximate $f_Z$ by compactly supported
		bounded pips. Find $g \in C_c$ such that $\|f_Z-g\|_1 \leq \epsilon$. Then $|\theta(x')-\theta(x)| \leq 4 \epsilon + \int_{\mreals^d} |g(z-x') -g(z)| \diff z - \int_{\mreals^d} |g(z-x) -g(z)| \diff z$. Send $x'\to x$ and then $\epsilon \to 0$.}
\end{itemize}
Further properties of $\theta$ in dimension $d=1$ include:
\begin{itemize}
	\item $\theta$ is continuous at 0 if and only if $Z$ has a density with respect to the Lebesgue measure. 
	To see this, decompose $P_Z = \mu_a + \mu_s$ 
	into absolutely continuous and singular parts (with respect
	to the Lebesgue measure). By~\cite[Theorem 10]{singular.translation}, 
	$\liminf_{h\to 0}\TV(P_Z, P_{Z+h})= 0$ if and only if $P_Z$ is absolutely continuous. By the previous remark we have
\[
			\limsup_{x\to 0} \theta(x) = \mu_s(\mreals).
\]

	\item If $P_Z$ has a non-increasing density supported on $\mreals_+$, then $\theta(x)$ is a concave,
		non-decreasing function on $\mreals_+$ given by
		\begin{equation}\label{eq:conc0}
				\theta(x) = \PP\left[ Z \le x \right]\,, \qquad x\ge0\,.
		\end{equation}			
	\item If $P_Z$ has a symmetric density which is non-increasing on $\mreals_+$, then $\theta(x)$ is a concave,
	non-decreasing function on $\mreals_+$ given by
	\begin{equation}\label{eq:conc1}
			\theta(x) = \PP\left[|Z| \le x / 2\right]\,, \qquad x \ge 0 
	\end{equation}		
	\item In general, $\theta$ need not be monotonic on $\mreals^+$ (e.g. $P_Z$ is discrete or has a multimodal
	density such as a Gaussian mixture).
\end{itemize}

\smallskip
The following result gives a necessary and sufficient condition for the total variation to strictly contract on an additive-noise channel, 
which essentially means that the noise distribution is almost mutually singular to a translate of itself. 
Intuitively, it means that if the noise is too weak (\eg, 
when the noise has a compact support or has a singular distribution), then
one can send one bit error-free if the signal magnitude is sufficiently large. 
\begin{theorem}\label{th:ftv_contracts} 
 Define 
	\[
	\eta(A) = \sup_{x:|x|\le A}\theta(x).
	\] 
The following are equivalent
	\begin{enumerate}
		\item $\eta(A)=1$ for some $A>0$.
		\item $\FTV(t)=t$ in some neighborhood of $0$.
		\item $\FTV(t)=t$ for some $t>0$.
	\end{enumerate}
\end{theorem}
\begin{remark}
	It is possible to have $\eta(A)=1$ with $\theta(x)<1$ on $[-A,A]$. For example, let 
	\[
	P_Z = \frac{1}{2} \sum_{k \geq 1} 2^{-k}\delta_k + \frac{1}{2} \sum_{k\geq 1} 2^{-k} U(2k-1,2k).
	\]
	where $U(a,b)$ denotes the uniform distribution on $(a,b)$. 
\end{remark}
\apxonly{\textbf{TODO:} Since characteristic function $\phi_Z$ in the vicinity of zero contains information about the tails of
	$P_Z$, can we deduce something about $\eta(A)$ from regularity of $\phi_Z$?}
\begin{proof} 
The	equivalence of $2$ and $3$ follows from Remark~\ref{rmk:FTVconcave}.

	For $1\Rightarrow 2$, choosing $P=(1-t)\delta_0 + t\delta_x$ and $Q=\delta_0$, we have $\TV(P*P_Z, Q*P_Z) = t \theta(x)$.
	Optimizing over $x \in [-A,A]$ yields $\FTV(t)=t$, provided that $t \cost(A)\le a$.

	Before proceeding further, we notice that for any 
	channel $P_{Y|X}$ with Dobrushin coefficient
	$\etaTV$ and any measure $\nu$ on $\mathcal{X}$ such that $\int \diff \nu = 0$ we have
	$$ \TV(P_{Y|X} \circ \nu, 0) \le \etaTV \TV(\nu, 0)\,, $$
	where here and below the total variation distance defined in 
	\prettyref{eq:tv} naturally extended to non-probability measures as follows: 
		$$ \TV(\nu, \mu)={1\over2} \int |\diff \nu - d\mu|\,.$$
	Next, by representing $\nu = \nu^+ - \nu^-$ and playing with
	scaling $\nu^+$ or $\nu^-$ we get the result of~\cite[Lemma 3.2]{CIR93}:
	$$\TV(P_{Y|X} \circ \nu, 0) \le \etaTV \TV(\nu, 0) + {1-\etaTV \over 2}\left|\int \diff \nu\right| $$ 
	
	Now we prove $3\Rightarrow 1$. Fix arbitrary $(P,Q)\in\matg_a$ and choose large $A>0$. 
	Let $P_1, Q_1$ be
	restrictions of $P$ and $Q$ to the closed ball
		$$ B \eqdef \{x: |x| \le A\} $$
	and $P_2 = P-P_1, Q_2=Q-Q_1$. 	
	By~\cite[Lemma 3.2]{CIR93} we have then:
	$$ \TV(P_1*P_Z, Q_1*P_Z) \le \eta \TV(P_1, Q_1) + {1-\eta \over 2} \left|P(B) - Q(B)\right|\,, \qquad
		\eta \eqdef \eta(A)\,.$$
		Since $(P,Q)\in\matg_a$, applying Markov's inequality yields $P(B^c) + Q(B^c) \le {2a\over \cost(A)}$
		and thus
		$$ \TV(P_2, Q_2) \le {a \over \cost(A)}. $$
	Also, since $P(\matx)-Q(\matx) = 0$, we have
	$$ \left|P(B) - Q(B)\right| = |P(B^c) - Q(B^c)| \le {2a \over \cost(A)}.  $$
	Putting it all together and using triangle inequality, we have
	\begin{align*} \TV(P*P_Z, Q*P_Z) &\le \TV(P_1*P_Z, Q_1*P_Z) + \TV(P_2 * P_Z, Q_2 * P_Z) \\
			&\le \TV(P_1*P_Z, Q_1*P_Z) + \TV(P_2, Q_2)\\
			&\le \eta \TV(P_1, Q_1) + {1-\eta\over2} \left|P(B) - Q(B)\right| + \TV(P_2, Q_2) \\
			&= \eta \TV(P, Q) + (1-\eta)\left({1\over2}\left|P(B) - Q(B)\right| + \TV(P_2,
			Q_2)\right)\label{eq:ttx1}\\
			&\le \eta \TV(P,Q) + (1-\eta) {2a\over \cost(A)}\,,
		\end{align*}			
	where the equality step follows from 
	the crucial fact that $\TV(P, Q) = \TV(P_1,Q_1)+\TV(P_2,Q_2)$, due to the disjointedness of supports.

	By the arbitrariness of $(P,Q)$, we have shown that for every $A>0$ and $t$,
	$$  \FTV(t) \le \eta(A) t  + (1-\eta(A)) {2a\over \cost(A)}\,.$$
	Thus if $\FTV(t)=t$ for some $t>0$, then $ (1-\eta(A)) t \le (1-\eta(A)) {2a\over \cost(A)} $ for all $A>0$.
	Therefore we must have $\eta(A)=1$ whenever $\cost(A)>{2a\over t}$.
	\end{proof}

\subsection{Bounds on $\FTV$ via coupling}
\label{sec:coupling}

\begin{theorem}\label{th:coupling} 
Define $\thetalb(s) = \sup_{x:|x| \leq s} \theta(x)$ and let $\theta_c: \reals_+ \to [0,1]$ be the concave envelope (\ie, the smallest concave majorant)
	of $\thetalb$ on $\mreals_+$, 
		Then
\begin{equation}\label{eq:coupling}
t \thetalb\left(2 \cost^{-1}\pth{a\over t}\right)  \le \FTV(t) \le t \theta_c\left(2 \cost^{-1}\pth{a\over t}\right) 
\end{equation}	
\end{theorem}
\begin{remark} Note that for the upper bound~\eqref{eq:coupling} to be non-trivial, \ie, better than $\FTV(t)\le t$, for all
	$t>0$, it is necessary and sufficient to have $\theta_c(|x|)<1$ for all $x$. This is consistent with Theorem~\ref{th:ftv_contracts}. 
\end{remark}
\begin{proof} 
Recall that $\TV(P_{Z+a},P_{Z+b})=\theta(a-b)$, by definition of the function $\theta$ in \prettyref{eq:theta}.
Fix any $(P,Q)\in\matg_a$.
The map $(P,Q)\mapsto \TV(P,Q)$ is convex (as is any Wasserstein distance), thus for any coupling $P_{AB}$ with $P_A=P$ and $P_B=Q$ we have
	\begin{equation}\label{eq:tc1}
			\TV(P*P_Z, Q*P_Z) \le \EE[\theta(A-B)]
\end{equation}	
	Furthermore, $\theta_c$ is necessarily continuous on $(0,\infty)$, strictly increasing on $\{x: \theta_c(x) <
	1\}$ and concave. Thus,
		\begin{align} \EE[\theta(|A-B|)] =& ~\PP[A\neq B] \EE[\theta(A-B) \,|\, A\neq B] \\
			     \le & ~\PP[A\neq B] \EE[\theta_c(|A-B|) \,|\, A\neq B]\\
			     \le & ~\PP[A\neq B] \theta_c\left(\EE[|A-B| \,|\, A\neq B]\right)\label{eq:tc3}
		     \end{align}
where \prettyref{eq:tc3} is by Jensen's inequality and the concavity of $\theta_c$. Then
\begin{align}
\cost\pth{\frac{\Expect[|A-B||A\neq B]}{2}}
\leq & ~ \EE\qth{\cost\pth{\frac{|A-B|}{2}} \Big| A\neq B} \label{eq:M1}\\
= &~ \frac{1}{\prob{A \neq B}} \EE\qth{\cost\pth{\frac{|A-B|}{2}}}	\label{eq:M2}\\
\leq & ~ \frac{1}{\prob{A \neq B}} \EE\qth{\cost\pth{\frac{|A|+|B|}{2}}}	\label{eq:M3}\\
\leq & ~ \frac{\EE[\cost(|A|)]+\EE[\cost(|B|)]}{2\prob{A \neq B}} \label{eq:M4}\\
\leq & ~ \frac{a}{\prob{A \neq B}} \label{eq:M5}
\end{align} 
where \prettyref{eq:M1} and \prettyref{eq:M4} are by Jensen's inequality and the convexity of $\cost$, \prettyref{eq:M2} is by $\cost(0)=0$, \prettyref{eq:M3} is by the monotonicity of $\cost$, and 
\prettyref{eq:M5} is by the constraint $(P,Q)\in \calG_a$.
Applying $\cost^{-1}$ to both sides of \prettyref{eq:M5} and plugging into \prettyref{eq:tc3}, we obtain
		\begin{align} 
		\EE[\theta(|A-B|)] \le \PP[A\neq B] \theta_c\left(2 \cost^{-1} \pth{\frac{a}{\prob{A \neq B}}}\right). \label{eq:tc2}
		     \end{align}
Note that both $\cost^{-1}$ and $\theta_c$ are increasing concave functions. Thus their composition
$\theta_c\circ2 \cost^{-1}$ is concave and increasing too. Furthermore it is easy to show that 
\begin{equation}\label{eq:tpersp}
		t \mapsto t \theta_c\left(2 \cost^{-1} \pth{\frac{a}{t}}\right)  
\end{equation}
is increasing. Hence the upper bound~\eqref{eq:tc2} is tightest for the coupling minimizing
$\PP[A\neq B]$.  Recall that by Strassen's characterization~\cite{strassen.marginal} we have
     \begin{equation}\label{eq:tc4}
	     	\inf_{P_{AB}} \PP[A\neq B] = \TV(P,Q),
\end{equation}	
where the infimum is over all couplings $P_{AB}$ of $P$ and $Q$ such that $P_A=P$ and $P_B=Q$. Then~\eqref{eq:tc1} and~\eqref{eq:tc2} and the continuity of
	$\theta_c$ imply the upper bound in~\eqref{eq:coupling}.

For the lower bound, we choose 
\begin{align} P &= (1-t) \delta_0 + t \delta_{x} \label{eq:tc5a}\\
	Q &= (1-t) \delta_0 + t \delta_{-x} \label{eq:tc5b} 
\end{align}	
with $|x| \leq \cost^{-1}(a/t)$, which ensures that $(P,Q)\in\calG_a$.
It is straightforward to show that $\TV(P,Q)=t$ and $\TV(P*P_Z,Q*P_Z) = t \theta(x)$.
Taking the supremum over $x$ yields the left inequality of \prettyref{eq:coupling}.
\end{proof}

\apxonly{\textbf{Remark:} The proof of Theorem~\ref{th:coupling} shows something more than stated. This extension is in
the spirit of~\cite{RLD70}. Let $W_\rho(P,Q)$ be Wasserstein distance with respect to some translation invariant metric
$\rho(\cdot, \cdot)$. Let
$$ \theta(x) = W_\rho(P_Z, P_{Z+x}) = W_\rho(P_Z, P_{Z-x})$$
and $\theta_c$ be its concave envelope on $\mreals_+$. Now consider two Markov chains as in~\eqref{eq:mc1}, started at
$P_{X_1}= P$ and $Q_{X_1}=Q$:
\begin{align}
	X_1 &\to Y_1 \to X_2 \to Y_2 \to \cdots \to X_n \to Y_n \to \cdots \\
	X'_1 &\to Y'_1 \to X'_2 \to Y'_2 \to \cdots \to X'_n \to Y'_n \to \cdots
\end{align}
And suppose that there already exist a coupling between the two chains satisfying (for each $k\le n$):
$$ \EE[\rho(X_k, X_k')] \le \alpha_k $$
Then the coupling can be extended to $k=n+1$ so that
\begin{equation}\label{eq:mcx}
		\EE[\rho(Y_{n+1}, Y'_{n+1})] \le f(\alpha_{n+1}), \quad f(t) = t \theta_c\left(2 \pth{a\over t}^{\frac{1}{p}}\right)
\,.
\end{equation}
When $\rho(a,b) = {1\over 2}1\{a\neq b\}$ this yields total variation, so that the estimate~\eqref{eq:mcx} implies
$$ \EE[\rho(X_{n+1}, X'_{n+1})] \le f(\alpha_{n+1}) $$
and hence the induction can be continued to produce coupling between $(X^\infty, Y^\infty, (X')^\infty, (Y')^\infty)$.
}

\begin{corollary} If the dimension $d=1$ and $\thetalb$ is concave on $\mreals_+$ then 
\begin{equation}
	\FTV(t) = t \thetalb\left(2 \costi\left(a\over t\right)\right)\,.  \label{eq:Ftv-Z}
\end{equation}
	\label{cor:Ftv-Z}
\end{corollary}

\begin{remark} 
	Examples of the noise distributions satisfying assumptions of \prettyref{cor:Ftv-Z} are given
	by~\eqref{eq:conc0} and~\eqref{eq:conc1}. Note that from concavity of $\theta$ the map
	\begin{equation}\label{eq:mt1}
		u \mapsto \thetalb(2\costi(u))
	\end{equation}		
	is also concave. Therefore, the map
	$$ (a,t) \mapsto t \, \thetalb\pth{2 \costi\left(a\over t\right)} $$ 
	is the \emph{perspective} of the concave
	function~\eqref{eq:mt1}, and hence is concave on $\mreals_+^2$~\cite[p. 161]{HUL96}. Consequently, for fixed $a>0$,
	$\FTV$ is concave, which, as we mentioned, does not immediately follow from the definition of $\FTV$.  \label{rmk:tau}\end{remark}

For the purpose of showing Theorem~\ref{th:main} we next point out the particularization of \prettyref{cor:Ftv-Z} to the AWGN channel. 
A representative plot of the $\FTV$ for the AWGN channel and average power constraint (second-order moment) is given in Fig.~\ref{fig:tvawgn}, which turns out to be dimension-independent.

\begin{corollary}[Vector Gaussian]
	Let $d \in \naturals \cup \{\infty\}$, $P_Z=\calN(0,\mathbf{I}_d)$ and $|x|=(\sum_{i=1}^d
	x_i^2)^{1/2}$ be the Euclidean norm. Then
\begin{equation}
	\FTV(t)	=  t \pth{1 - 2 \sfQ \pth{ \costi\pth{\frac{a}{t}}}}.
	\label{eq:Ftv-g}
\end{equation}
	\label{cor:Ftv-g}
\end{corollary}
\begin{proof} From~\eqref{eq:tvnorm} we have that $\thetalb(u)=\thetac(u) = 1-2\sfQ(u/2)$ regardless of dimension and thus
	Theorem~\ref{th:coupling} yields~\eqref{eq:Ftv-g}.
\end{proof}

\begin{figure}[t]
	\centering
\ifpdf
	\includegraphics[width=.4\textwidth]{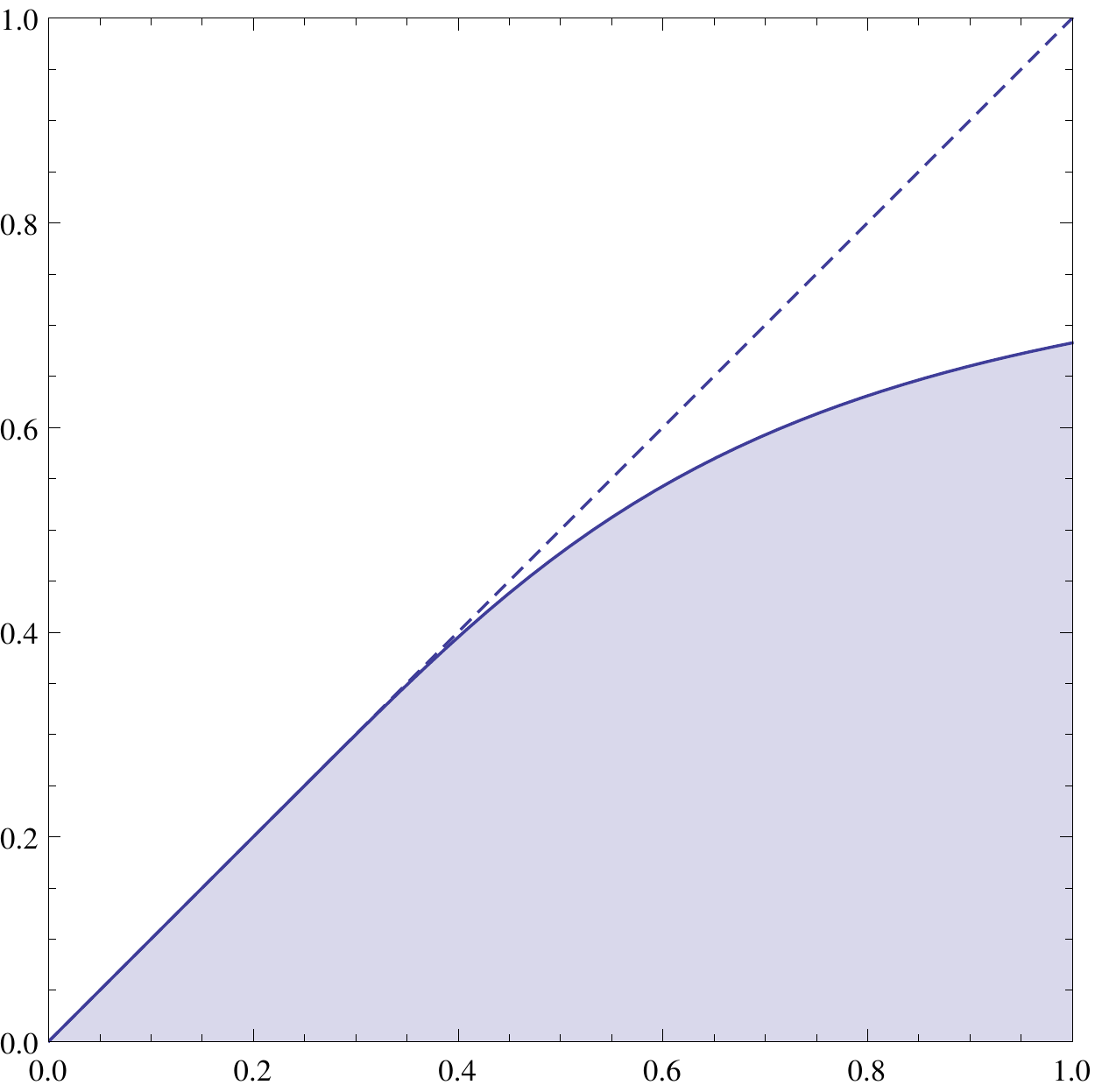}
	\put (-50,105) {$\FTV(t)$}
\else
	\includegraphics[width=.4\textwidth]{tv-awgn.eps}
\fi
	\caption{The region $\{(\TV(P,Q), \TV(P*P_Z, Q*P_Z)): (P,Q) \in \calG_1 \}$ in the Gaussian case $Z \sim
	\calN(0,1)$ with $\cost(x)=|x|^2$.}
	\label{fig:tvawgn}
\end{figure}

\subsection{By-product: CLT in smoothed total variation}
\label{sec:clt}

Recall the following 1-Wasserstein distance between distributions with finite first moment:
\begin{equation}
	W_1(P,Q) = \inf_{P_{AB}} \{\|A-B\|_1: P_A=P,P_B=Q\}.
	\label{eq:w1}
\end{equation}
Then the same coupling method in the proof of \prettyref{th:coupling} yields the following bound, which relates the total variation between convolutions to the $W_1$ distance.
\begin{prop} 
If $P_Z$ has a symmetric density which is non-increasing on $\mreals_+$. Then for any $P$ and $Q$, 
	\begin{equation}
\TV(P * P_Z, Q * P_Z) \leq \prob{|Z| \leq \frac{W_1(P, Q)}{2}}.
	\label{eq:tv-w1}
\end{equation}
	\label{prop:tv-w1}
\end{prop}

\begin{proof}
By \prettyref{eq:conc1}, the function $\theta(x) = \Prob[|Z|\leq x/2]$ is concave and non-decreasing in $x$. Applying Jensen's inequality to \prettyref{eq:tc1} and optimizing over the coupling yields \prettyref{eq:tv-w1}.
\end{proof}

\begin{remark}
	It is worth mentioning that for Gaussian smoothing, using similar coupling and convexity arguments,
the following counterpart of \prettyref{eq:tv-w1} for KL divergence has been proved in \cite{hwi},
which provides a simple proof of Otto-Villani's HWI inequality \cite{OV00} in the Gaussian case:
\[
D(P * \calN(0,\sigma^2) \| Q * \calN(0,\sigma^2)) \leq \frac{W_2^2(P,Q)}{2\sigma^2},
\]
where the $W_2$ distance is analogously defined as \prettyref{eq:w1} with $L_2$-norm replacing $L_1$-norm.
\end{remark}


In particular, if $P_Z$ has a bounded density near zero, then the right-hand side of \prettyref{eq:tv-w1} is $O(W_1(P,Q))$. 
As an application, we consider a central limit theorem setting and let
$$ S_n = {1\over \sqrt{n}} \sum_{j=1}^n X_j, $$
where $X_j$ are iid, zero-mean and unit-variance.
Choosing $P_Z = \matn(0, \sigma^2)$ and applying \prettyref{prop:tv-w1} to $P_{S_n}$ and $\calN(0,1)$, we obtain
\begin{equation}
\TV( P_{S_n} *\matn(0,\sigma^2), \matn(0, 1+\sigma^2)) \le \frac{W_1(P_{S_n},\calN(0,1))}{\sqrt{2\pi \sigma^2}} \leq {3 \, \EE[|X_1|^3] \over \sqrt{2 \pi \sigma^2 n}}		
	\label{eq:smooth-clt}
\end{equation}
where the convergence rate in $W_1$ can be obtained from Stein's method and the dual representation of $W_1(P,Q) = \sup\{\int f\diff P-\int f\diff Q: f \text{ is 1-Lipschitz}\}$ (see, \eg, \cite[Theorem 3.2]{BC05}). In other words, smoothing the law of $S_n$ by convolving with a Gaussian density (or any other bounded density that satisfies the conditions of \eqref{eq:conc1}) results in a distribution that is closer in total variation to the Gaussian distribution. On the other hand, the law of $S_n$ might never converge to Gaussian (\eg, for discrete $X_1$).

The non-asymptotic estimate \prettyref{eq:smooth-clt} should be contrasted with the sharp asymptotics of total variation in CLT due to Sirazhdinov and Mamatov \cite{SM62}, which states that the left-hand side of \prettyref{eq:smooth-clt} is equal to $\frac{(1+4\eexp^{-3/2})\Expect[X_1^3]}{6\sqrt{2 \pi n (1+\sigma^2)^3}} (1+o(1))$ when $n\diverge$ and $\sigma$ is fixed.

\apxonly{note to apply \cite{SM62} we need to normalize to unit variance by $1/\sqrt{1+\sigma^2}$. absorb convolution with $\calN(0,\sigma^2)$ into the normalizing sum. note that $X_1+Z_1$ has a density and $\Expect[(X+\sigma Z)^3] = \Expect[X^3]$.}






\section{From total variation to $f$-divergences}
\label{sec:mate}

The main apparatus for obtaining the Dobrushin curve of total variation in \prettyref{th:coupling} is the infimum-representation via couplings, 
thanks to the special role of the total variation as a Wasserstein distance.
Unfortunately such representation is not known for other divergences such as the Hellinger distance or KL divergence.
To extend the contraction property of total variation, our strategy is as follows: We first study a special family of $f$-divergences $\{\mate_{\gamma}(P\|Q): \gamma > 0\}$,
 which enjoys the same contraction property as the total variation for any channel. Then using an integral representation of general $f$-divergences \cite{CKZ98} in terms of $\mate_{\gamma}$, we extend the contraction results in \prettyref{sec:coupling} for additive-noise channels to $f$-divergences, in particular, R\'enyi divergences.

\subsection{A parameterized family of $f$-divergences}
For a pair of distributions $P,Q$, define the following family of $f$-divergences parameterized by $\gamma \geq 0$:
\begin{equation}
 \mate_{\gamma}( P\| Q) = \frac{1}{2} \int |\diff P - \gamma \diff Q| - \frac{1}{2} |1-\gamma|\,.
	\label{eq:mate}
\end{equation}
Typical plots of $\gamma \mapsto \mate_{\gamma}( P\| Q)$ are given in \prettyref{fig:mate} where $P$ and $Q$ are Gaussians or Bernoullis.

\begin{figure}[ht]
	\centering
	\ifpdf
		\subfigure[$P=\calN(1,1)$, $Q=\calN(0,1)$.]%
	{\label{fig:mate-g} \includegraphics[width=.4\columnwidth]{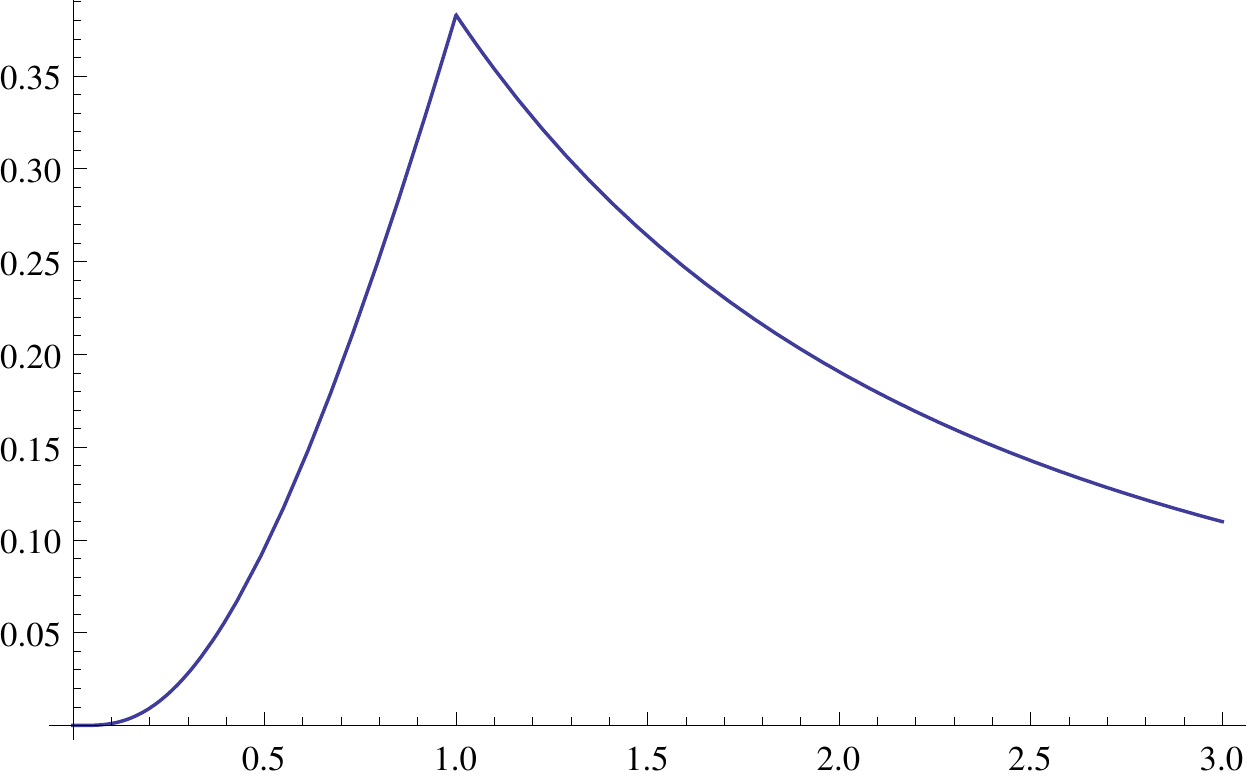}}
	\subfigure[$P=\text{Bern}(0.5)$, $Q=\text{Bern}(0.1)$.]%
	{\label{fig:mate-b} \includegraphics[width=.4\columnwidth]{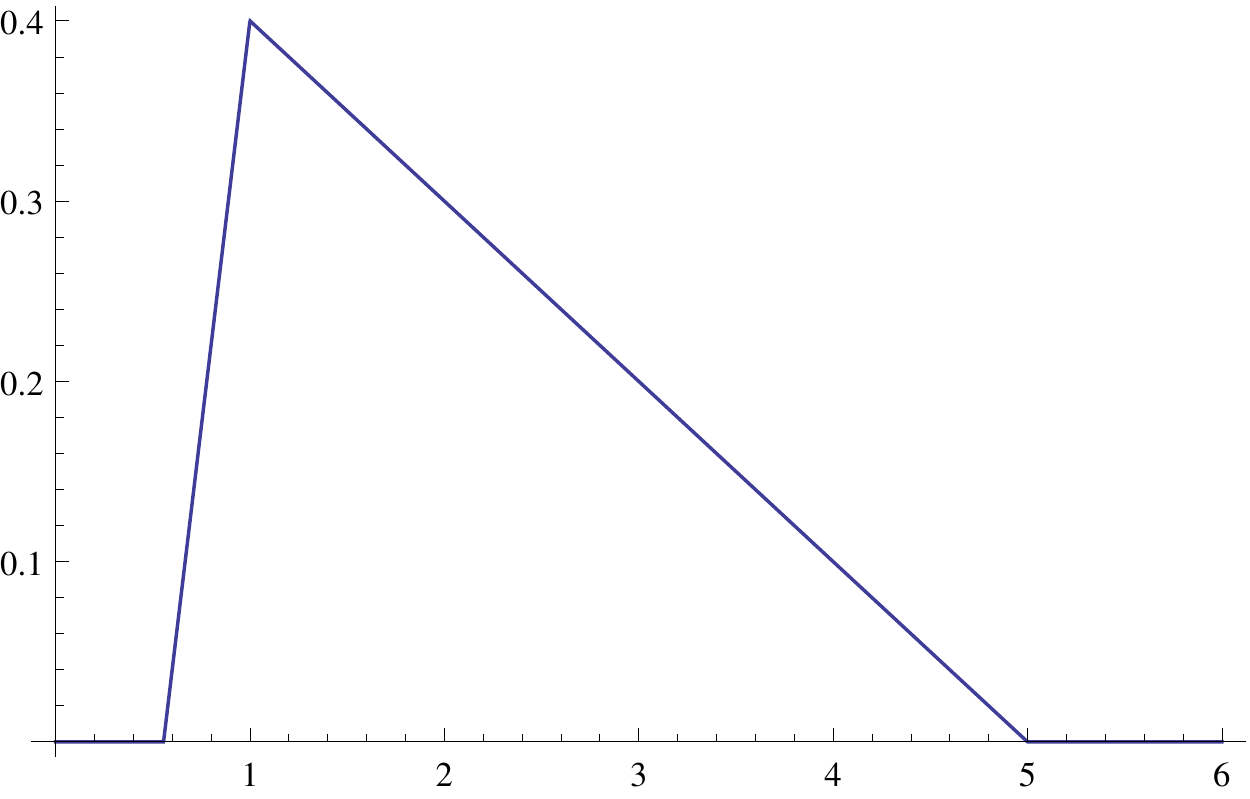}}
\else
		\subfigure[]%
	{\label{fig:mate-g} \includegraphics[width=.4\columnwidth]{mate.eps}}
	\subfigure[$P=\text{Bern}(0.5)$, $Q=\text{Bern}(0.1)$.]%
	{\label{fig:mate-b} \includegraphics[width=.4\columnwidth]{mate-binary.eps}}
\fi
	\caption{Plot of $\gamma \mapsto \mate_{\gamma}( P\|Q)$.}
	\label{fig:mate}
\end{figure}

Some general properties of $\mate_\gamma$ are as follows:
\begin{enumerate}
	\item $\mate_0(P\|Q)=\mate_\infty(P\|Q)=0$.
	\item $\mate_1(P\|Q)=\TV(P,Q)$.
	\item $\gamma \mapsto \mate_\gamma(P\|Q)$ is convex, positive, increasing on $[0,1]$, and decreasing on $[1,+\infty)$.
	\item Reciprocity:
	$$ \mate_{\gamma}( P\| Q) = \gamma \mate_{\gamma^{-1}}( Q\| P). $$
	\item Derivative of $\mate$ recovers $Q[\fracd{P}{Q} < \gamma]$: 
	\begin{align} {d\mate_\gamma\over d\gamma} 
	&= \frac{1}{2} \sign(1-\gamma) + Q\Big[\fracd{P}{Q} < \gamma\Big] - \frac{1}{2} \nonumber \\
		&= \indc{\gamma < 1} - Q\Big[\fracd{P}{Q} > \gamma\Big]  \label{eq:dmate}
	\end{align}		
		\item $F$-contraction property: If $P', Q'$ are outputs of $P,Q$ under some channel
		$P_{Y|X}$ with known $\FTV$, then 
		\begin{equation}\label{eq:mate_contracts}
				\mate_{\gamma}( P'\| Q') \le \FTV(\mate_{\gamma}( P\| Q)).
		\end{equation}			
		This follows from the more general result below, which shows that the divergence $\mate_\gamma$ for general $\gamma$ enjoys the same (if not better) contraction property as the total variation, \ie, $\mate_1$.
\end{enumerate}

\begin{proposition}
 Assume that for each choice of $a>0$ in~\eqref{eq:Ga} the corresponding $\FTV$ curve is denoted by 
	$t \mapsto \FTV(t,a)$. Then for any channel $P_{Y|X}$ and any $(P,Q)\in \calG_a$ we have
	\begin{equation}\label{eq:mec}
			\mate_{\gamma}( P_{Y|X} \circ P\| P_{Y|X} \circ Q) \le \FTV(\mate_{\gamma}( P\| Q), a(\gamma \wedge 1))
\end{equation}	
	and, in particular, \eqref{eq:mate_contracts} holds.
	\label{prop:mate}
\end{proposition}
\begin{proof} First notice that if $\nu$ is any signed measure on $\calX$ satisfying
\begin{equation}
\int \diff \nu = 0, \quad \int \diff |\nu| \le 2, \quad \int \cost(|x|) \,\diff |\nu| \le 2a'.
	\label{eq:signed}
\end{equation}
for some $a'>0$, then we have\footnote{The push-forward operation is extended to signed non-probability measures in the obvious way: $P_{Y|X} \circ \nu(\cdot) = \int P_{Y|X}(\cdot) \nu^+(\diff x) - \int P_{Y|X}(\cdot) \nu^-(\diff x)$.}
	\begin{equation}\label{eq:rx1}
			\TV(P_{Y|X} \circ \nu, 0) \le \FTV(\TV(\nu, 0), a') .
	\end{equation}	
	Indeed, let $\nu = \nu^+ - \nu^-$ be the Jordan decomposition of $\nu$. Then by the assumption \prettyref{eq:signed} we have that $\nu^{\pm}$ are mutually singular 
	sub-probability measures. 
	Thus by introducing $P=\nu^+ + \epsilon \delta_0, Q=\nu^- + \epsilon \delta_0$ for some constant
	$\epsilon\ge 0$ chosen so that $P$ and $Q$ are probability measures, we get
	\begin{equation}\label{eq:rx2}
			\TV(P_{Y|X} \circ P, P_{Y|X} \circ Q) \le \FTV(\TV(P,Q), a') 
	\end{equation}	
	since $(P,Q)\in\calG_{a'}$. In turn,~\eqref{eq:rx2} is equivalent to~\eqref{eq:rx1}.

	Now consider $\gamma<1$ and a pair of probability measures $(P,Q)\in\calG_a$. Write $\mate_\gamma = \mate_{\gamma}( P\| Q)$ and set
	$$ \nu = c (P-\gamma Q)^+ - (P-\gamma Q)^-, \qquad c = {\mate_\gamma\over \mate_\gamma + 1-\gamma}\leq 1\,.$$
	Since $\mate_\gamma \le \gamma \TV(P,Q) \le \gamma$, which follows from the convexity of $\gamma \mapsto \mate_\gamma$, we have $c\le \gamma$. Then
		$$ \int \cost(|x|) \,\diff |\nu| \le \gamma \EE_Q[\cost(|X|)] + c \, \EE_P[\cost(|X|)] \le 2\gamma a\,.$$
	Consequently $\nu$ satisfies condition \prettyref{eq:signed} with $a'=\gamma a$. Furthermore, observe that  for $\gamma \leq 1$ we have
	\begin{equation}
	\int (\diff P-\gamma \diff Q)^- = \mate_\gamma, \quad \int (\diff P-\gamma \diff Q)^+ = \mate_\gamma+1-\gamma, 
	\label{eq:dPQpm}
\end{equation}
	we have
	$\TV(\nu, 0) = \mate_\gamma$. Thus from~\eqref{eq:rx1} we get
	$$ \TV(P_{Y|X} \circ \nu, 0) \le \FTV(\mate_\gamma, \gamma a) .$$
	Next from the representation
	$$ P-Q = \nu + (1-c)(P-\gamma Q)^+ $$
	and the triangle inequality we have
	\begin{equation}\label{eq:rx3}
			\TV(P_{Y|X} \circ P, P_{Y|X} \circ Q) \le \TV(P_{Y|X} \circ \nu, 0) + {1-c \over 2} \int (\diff P-\gamma \diff Q)^+\,. 
	\end{equation}	
	In view of \prettyref{eq:dPQpm}, 
	it remains to notice that the last term in~\eqref{eq:rx3} equals $1-\gamma\over 2$, from which~\eqref{eq:mec}
	follows via 
	$$ \mate_{\gamma}( P_{Y|X} \circ P\| P_{Y|X} \circ P) = \TV(P_{Y|X} \circ P, P_{Y|X} \circ Q) - {1-\gamma\over 2}\,. $$

	For $\gamma > 1$ the proof is entirely analogous, except that we set
	$$ \nu = (P-\gamma Q)^+ - c(P-\gamma Q)^-\,, \qquad c = {\mate \over \mate + \gamma -1} \le 1 $$
	and the best bound we have on $\int \cost(|x|) \diff |\nu|$ is $2a$, which follows from the fact that $\mate_\gamma \leq 1$ and hence $c\gamma \le 1$.
		\apxonly{and in this case \prettyref{eq:dPQpm} changes to $\int (\diff P-\gamma \diff Q)^+ = \mate_\gamma, \quad \int (\diff P-\gamma \diff Q)^+ = \mate_\gamma+\gamma-1$.}
\end{proof}

\subsection{Integral representation and contraction of R\'enyi divergences}

For an $f$-divergence, analogous to the Dobrushin curve \prettyref{eq:Ft} we define
\begin{equation}
F_{f}(t) \triangleq \sup\{ D_f(P_{Y|X} \circ P, P_{Y|X} \circ Q): D_{f}(P, Q) \le t, (P,Q) \in  \calG_a\},  
	\label{eq:FKLt}
\end{equation}
Note that the usual data processing inequality amounts to $F_f(t) \leq t$. We say the channel $P_{Y|X}$ \emph{contracts} the $f$-divergence $D_f$ if $F_f(t) < t$ for all $t$ in a neighborhood near zero. We have already shown that the total variation is always contracted by additive noise satisfying the necessary and sufficient condition in \prettyref{th:ftv_contracts}. In view of \prettyref{prop:mate}, the formulas in Corollaries \ref{cor:Ftv-Z} and \ref{cor:Ftv-g} apply to $\mate_\gamma$ as well. A natural question is in order: Do other divergences, such as the KL divergence, also contract in additive noise? To this end, we need the following integral representation of $f$-divergences in terms of the family of divergence $\mate_\gamma$: If $f \in C^2(\reals_+)$, then (see \cite[Corollary 3.7, p. 99]{CKZ98})
\begin{equation}
	D_f(P||Q) = \int_0^\infty \mate_\gamma(P\|Q) f''(\gamma) \diff \gamma.
	\label{eq:mateint}
\end{equation}
For instance, the area under the curve $\gamma \mapsto \mate_\gamma$ is half the $\chi^2$-divergence $\chi^2(P\|Q)=\int \frac{\diff P^2}{\diff Q}-1$.
\apxonly{
Representation of KL:
	$$ D(P||Q) = \log e \cdot \int_0^{\infty} \frac{\mate_\gamma}{\gamma} \diff \gamma $$
}

For conciseness, below we focus on the scalar AWGN channel under the first moment constraint and the special case of R\'enyi divergence of order $\alpha$, which is a monotonic transformation of the $f_\alpha$-divergence with 
\begin{equation}
f_\alpha(x)=
\begin{cases}
1 - x^{\alpha} & \alpha \in (0,1)\\
x \log x & \alpha=1\\
x^{\alpha}-1 & \alpha>1
\end{cases}.
	\label{eq:falpha}
\end{equation}
Note that the special case of $\alpha = 1,2,\frac{1}{2}$ corresponds to the KL divergence $D(P\|Q)$, the $\chi^2$-divergence $\chi^2(P\|Q)$, and half the squared Hellinger distance $H^2(P,Q)=\int (\sqrt{\diff P}-\sqrt{\diff Q})^2$, respectively.
The following result shows that the AWGN channel contracts R\'enyi divergence of order $\alpha$ if and only if $\alpha \in (0,1)$. Consequently, the Hellinger distance always contracts when passing through the AWGN channel, but $\chi^2$ and KL divergences do not. 

\begin{theorem}
Consider the scalar AWGN channel $P_{Y|X}=\calN(X,1)$. Let $\cost(|x|)=|x|^2$ and $a>0$.
 Then
\begin{enumerate}
	\item For $\alpha \in (0,1)$, for any $\epsilon>0$,
	\begin{equation}
F_{f_\alpha}(t) \leq t \pth{1 - 2 \sfQ \pth{\sqrt{a} t^{-\frac{1+\epsilon}{2\alpha}}}} +  (1+\alpha-\alpha^2) t^{1+\epsilon}, \quad 0<t<1.
	\label{eq:Falpha-1}
\end{equation}
	\item For $\alpha \geq 1$, 
	\begin{equation}
	F_{f_\alpha}(t) = t,
	\label{eq:Falpha-2}
\end{equation}
which holds for all $t>0$ if $\alpha > 1$ and $t < \frac{a}{8}$ if $\alpha=1$, respectively.
\end{enumerate}	
	\label{thm:Falpha}
\end{theorem}
\begin{proof}
$1^\circ$	Fix $\alpha \in (0,1)$ and $(P,Q) \in \calG_a$ such that $D_{f_\alpha}(P\|Q) \in (0,1)$. Let $\calN$ denote the standard normal distribution. Fix $\delta>0$. Applying the integral representation \prettyref{eq:mateint} to $f_\alpha(x)=1 - x^{\alpha}$, we have
\begin{align}
D_{f_\alpha}(P*\calN\|Q*\calN)
= & ~ \alpha(1-\alpha) \int_0^\infty \mate_\gamma(P*\calN\|Q*\calN) \gamma^{\alpha-2}  \diff \gamma		\nonumber \\
\leq & ~ \alpha(1-\alpha) \int_0^\infty \mate_\gamma \pth{1 - 2 \sfQ \pth{\sqrt{\frac{a}{\mate_\gamma }}}} \gamma^{\alpha-2}  \diff \gamma	\label{eq:int1}\\
\leq & ~ \pth{1 - 2 \sfQ \pth{\sqrt{\frac{a}{\delta}}}} D_{f_\alpha}(P\|Q) + \alpha(1-\alpha) \int_0^\infty \mate_\gamma \gamma^{\alpha-2} \indc{\mate_\gamma \leq \delta}  \diff \gamma	\label{eq:int2} \\
\leq & ~ \pth{1 - 2 \sfQ \pth{\sqrt{\frac{a}{\delta}}}} D_{f_\alpha}(P\|Q) + \alpha(1-\alpha) \int_0^1 \mate_\gamma \gamma^{\alpha-2} \indc{\mate_\gamma \leq \delta}  \diff \gamma + \alpha\delta	,\label{eq:int3}
\end{align}
where 
\prettyref{eq:int1} follows from \prettyref{cor:Ftv-g} with $\mate_\gamma=\mate_\gamma(P\|Q)$,
and \prettyref{eq:int2} follows from \prettyref{eq:mateint}, and \prettyref{eq:int3} is due to $\mate_\gamma \leq \TV \leq 1$.
	Using \prettyref{eq:dmate}, for all $\gamma \in (0,1)$, we have $\mate_\gamma' = Q[\fracd{P}{Q} > \gamma] \leq 1$. By the convexity of $\gamma \mapsto \mate_\gamma$ and $\mate_0=0$, we have $\mate_\gamma \leq Q[\fracd{P}{Q} > \gamma] \gamma \leq \gamma$. Therefore
\begin{align}
\int_0^1 \mate_\gamma \gamma^{\alpha-2} \indc{\mate_\gamma \leq \delta}  \diff \gamma
= & ~ \int_0^\delta \mate_\gamma \gamma^{\alpha-2} \indc{\mate_\gamma \leq \delta}  \diff \gamma + \int_\delta^1 \mate_\gamma \gamma^{\alpha-2} \indc{\mate_\gamma \leq \delta}  \diff \gamma 	\nonumber \\
\leq & ~ \int_0^\delta \gamma^{\alpha-1} \diff \gamma + \delta^{\alpha-1} \int_\delta^1   \indc{\mate_\gamma \leq \delta} \mate_\gamma' \diff \gamma 	\nonumber \\
\leq & ~ \frac{1+\alpha}{\alpha} \delta^{\alpha} \label{eq:int4}.	
\end{align}
Plugging \prettyref{eq:int4} into \prettyref{eq:int3} and by the arbitrariness of $\delta>0$, we obtain
\[
D_{f_\alpha}(P*\calN\|Q*\calN) \leq \inf_{0<\delta<1} \sth{ \pth{1 - 2 \sfQ \pth{\sqrt{\frac{a}{\delta}}}} D_{f_\alpha}(P\|Q) + (1+\alpha-\alpha^2)\delta^{\alpha} },
\]
which implies the desired \prettyref{eq:Falpha-1} upon choosing $\delta = (D_{f_\alpha}(P\|Q))^{(1+\epsilon)/\alpha}$.
	
$2^\circ$	Turning to the case of $\alpha \geq 1$, we construct examples where $D_{f_\alpha}$ does not contract.
	Fix $t > 0$ and let $q>0$ be sufficiently small. Let $P_q = (1-p) \delta_0 + p \delta_b, Q_q = (1-q) \delta_0 + q \delta_b$ with $b = \sqrt{\frac{a}{p}}$ and $p = \frac{t}{\log \frac{1}{q}}$ if $\alpha = 1$ and $p = q(\frac{t}{q})^{1/\alpha}$ if $\alpha > 1$.
Then it is clear that $(P_q,Q_q) \in \calG_a$ for all sufficiently small $q$. Furthermore, 
$$ D_{f_\alpha}(P_q \| Q_q) = d_\alpha(p\|q) = t + o(1), \qquad q\to0\,,$$
where $d_{\alpha}(p\|q) \triangleq q^{1-\alpha} p^\alpha + (1-q)^{1-\alpha} (1-p)^\alpha$ if $\alpha>1$ and $p\log \frac{p}{q} + (1-p) \log \frac{1-p}{1-q}$ if $\alpha=1$.

Next, by applying the data-processing inequality to the transformation $y \mapsto \indc{y \geq b/2}$ we get
$$ D_{f_\alpha}(P_q * \matn \| Q_q * \matn) \ge d_\alpha(p'\|q') ,$$
where $p' = p + (1-2p) \sfQ(b/2) =  p(1+o(1))$ and $q' = q + (1-2q) \sfQ(b/2) = q(1+ o(1)$.
This follows from the fact that $\sfQ(b/2) = o(q)$, which is obvious for $\alpha > 1$; for $\alpha = 1$, since we have assumed that $t < a/8$, we have $\sfQ(b/2) \leq \exp(-b^2/8) = q^{a/8t} = o(q)$.
Consequently, 
$D_{f_\alpha}(P_q* \matn\|Q_q* \matn) \geq d_\alpha(p'\|q') = t + o(1)$ as $q\to0$, 
which completes the proof of \prettyref{eq:Falpha-2}. 
\end{proof}

\begin{remark}
\prettyref{thm:Falpha} extends in the following directions:
\begin{enumerate}
	\item For general additive noise $Z$, \prettyref{eq:Falpha-1} continues to hold with $1 - 2 \sfQ(\cdot)$ replaced by the concave envelope $\theta_c(\cdot)$ in \prettyref{th:coupling}.
	\item For the $p\Th$-moment constraint with $\cost(|x|)=|x|^p$ and $p>2$, \prettyref{eq:Falpha-2} holds for all $t,a>0$ if $\alpha > 1$. For KL divergence ($\alpha = 1$), however, it remains unclear whether \prettyref{eq:Falpha-2} holds in a neighborhood near zero since the above construction no longer applies.
	\apxonly{the only change now is to let $b = (a/\tau)^{1/p}$}
\end{enumerate}
\end{remark}

\section{Proof of Theorem~\ref{th:main}}\label{sec:proof}

Theorem~\ref{th:main} follows from Propositions~\ref{prop:tv-lost},~\ref{prop:iconv} and~\ref{prop:rconv} given in
Sections~\ref{sec:T},~\ref{sec:I} and~\ref{sec:rho}, respectively. The special case of finite-alphabet $W$ is much
simpler and is treated by Proposition~\ref{prop:ITV} (Section~\ref{sec:discrete}). Finally, \prettyref{sec:ach} shows that our
converse bounds are optimal for total variation, mutual information and
correlation in the scalar Gaussian case.

\subsection{Convergence in total variation}\label{sec:T}

The development in \prettyref{sec:curve} deals with comparing a pair of distributions and studies by how much their total variation shrinks due to smoothing by the additive noise. Therefore these results are applicable to binary sources, \ie, transmitting one bit. What if the sources takes more than two, or rather, a continuum of, values?
To this end, the data processing inequality for mutual information is relevant, which states that $W \to X \to Y \to Z$ implies that $I(W; Z) \leq I(X;Y)$. In other words, dependency decreases on Markov chains.
Our goal next is to find a quantitative data pre-processing and post-processing inequalities as a counterpart of
\prettyref{th:coupling}. Since we know, in view of \prettyref{thm:Falpha}, that KL divergence does not contract, it is
natural to turn to total variation and define the following \emph{$T$-information}:
\begin{equation}	
T(X;Y) \triangleq \TV(P_{XY}, P_XP_Y),
	\label{eq:TXY}
\end{equation}
which has been studied in, \eg, \cite{Csiszar96,Pinsker05}. 
Similar to mutual information, it is easy to see that the $T$-information satisfies the following properties:
\begin{enumerate}
	\item $T(X;Y) = \Expect[\TV(P_{Y|X}, P_Y)]=\Expect[ \TV(P_{X|Y}, P_X)]= T(Y;X)$.
	\item Data-processing inequality: $W \to X \to Y \to Z$ implies that $T(W; Z) \leq T(X;Y)$.
	\item If $S$ is Bern($1\over2$), then 
		\begin{equation}
	T(S;X) = \frac{1}{2} \TV(P_{X|S=0}, P_{X|S=1}) .
	\label{eq:Tonebit}
\end{equation}
\item If $S$ and $\hat S$ are both binary, then\footnote{To see this, let $S \sim \Bern(p), \hat S \sim \Bern(q)$, $p_0=\Prob[S = 1|\hat S=0]$ and $p_1=\Prob[S = 0|\hat S=1]$. Then $T(S;\hat S) = \bar q |p_0-p|+ q|p_1 - \bar p| \geq p \bar q + \bar p q - (\bar q p_0 + q p_1) \geq \min\{p,\bar p\} - \Prob[\hat S \neq S]$. 
\apxonly{More precisely, $T(S;\hat S) = 2 prob{S=0} \prob{S=1} |1 - \Prob[\hat S = 1|S=0] - \Prob[\hat S = 1|S=0]|$ and $T(S;S) = 2 prob{S=0} \prob{S=1}$.}}
\begin{equation}
	T(S;\hat S) \geq \min\{\prob{S = 0},\prob{S=1}\} - \pprob{S \neq \hat S}.
	\label{eq:T-binary}
\end{equation}
\item Pinsker's inequality:
\begin{equation}
I(X;Y) \geq 2 \log \eexp \,	T(X;Y)^2.
	\label{eq:pinsker}
\end{equation}
\end{enumerate}

The next theorem gives a quantitative data processing theorem for the $T$-information with additive noise:
\begin{theorem}\label{thm:TVproc}
	Let $W \to X \to Y$, where $Y = X+ Z$ and $\Expect [\cost(|X|)] \leq a$. Let $\theta_c$ be as in
	Theorem~\ref{th:coupling}.  Then	
\begin{equation}
	T(W;Y) \leq f(T(W;X), a), \quad f(t, a) \triangleq t \theta_c\left(2 \costi\pth{a\over t}\right)\,.
	\label{eq:TVproc}
\end{equation}
\end{theorem}
\begin{remark} Exactly the same inequality holds for the following functional of real-valued random
	variables
		$$ T'(A;B) \triangleq \inf_{\Expect_{Q_B}[\cost(|B|)] \le a} \TV(P_{AB}, P_A Q_B), $$
		which is a natural extension of the $K$-information of Sibson~\cite{RS69} and
		Csisz\'ar~\cite{IC95} and satisfies $T'(A;B)\leq T(A;B)$. 
		Optimizing over $Q_B$ instead of taking $Q_B=P_B$ may lead to more powerful converse bounds,
		see~\cite{PV10-ari} for details. 
\end{remark}
\begin{proof} By the definition of $T(W;Y)$ and the Markov chain condition, we have
	$$ T(W;Y) = \int\TV(P_{X|W=w} * P_Z, P_X * P_Z) \, P_W(\diff w).  $$
	Then \prettyref{th:coupling} yields
	\begin{equation}
			\TV(P_{X|W=w} * P_Z,P_X * P_Z) \leq f\left(\TV(P_{X|W=w},P_X),\, {1\over 2}\EE[\cost(|X|)|W=w] +
	{1\over2}\EE[\cost(|X|)] \right)\,.
	\label{eq:ttww}
	\end{equation}	
	In view of \prettyref{rmk:tau}, the function $f$ defined in \prettyref{eq:TVproc} is jointly concave and
	non-decreasing in each argument. Thus taking
	expectation over $w \sim P_W$ on the right-hand side of~\eqref{eq:ttww} and applying Jensen's inequality, we complete the proof.
\end{proof}

	As an application of \prettyref{thm:TVproc}, next we describe how the $T$-information decays on the Markov chain \prettyref{eq:mc1}.
\begin{proposition} \label{prop:tv-lost}
Assume the Markov chain \prettyref{eq:mc1}, where $Z_j$ are i.i.d. and $\Expect[\cost(X_j)] \leq a$ for all $j\in [n]$.
Then for all $a>0$ and $n\geq 2$,
\begin{equation}\label{eq:tvl}
	T(W; Y_n) \leq \frac{a}{f^{-1}(n-1)},
\end{equation}
where $f(s) \triangleq \int_1^s \frac{1}{y \, (1-\theta_c(2 \cost^{-1}(y)))}\diff y$.

	In particular, if $Z_j \sim \calN(0,1)$ are i.i.d., then
	\begin{equation}\label{eq:tvl-g}
	T(W; Y_n) \leq C a \exp( - g^{-1}(n)),
\end{equation}
where $g(s) \triangleq \int_0^s \exp[\frac{1}{2} \costi(\exp(\tau))^2] \diff \tau$, and $C$ is a positive constant only depending on the cost function $\cost$.
\end{proposition}

\begin{remark}[Gaussian noise]
Particularizing the result of \prettyref{prop:tv-lost} to the AWGN channel and the following cost functions we obtain the corresponding convergence rates 
\begin{enumerate}[a)]
	\item $p\Th$-moment constraint: $\Expect |X_k|^p \leq a$ for some $p\geq 1$. Then $T(W; Y_n) = O((\log n)^{-p/2})$. In particular, for power constraint $\cost(x)=x^2$, \prettyref{eq:main1} holds.
\item Sub-exponential: $\Expect \exp(\alpha |X_k|^2) \leq a$ for some $\alpha > 0$ and $a>1$. Then $T(W; Y_n) = O(\eexp^{-\sqrt{2\alpha \log n}})$.
\item Sub-Gaussian: $\Expect \exp(\alpha |X_k|^2) \leq a$ for some $\alpha > 0$ and $a>1$. Then $T(W; Y_n) = O(n^{-2\alpha})$.
\end{enumerate}
Intuitively, the faster the cost function grows, the closer we are to amplitude-constrained scenarios, where we know that information contracts linearly thanks to the Dobrushin's coefficient being strictly less than one. Hence we expect the convergence rate to be faster and closer to, but always strictly slower than, exponential decay.
	In view of \prettyref{eq:Tonebit}, \prettyref{prop:tv-lost} implies that transmitting one bit is impossible
	under any cost constraint, since the optimal Type-I+II error probability is given by
	${1\over2}-\TV(P_{Y_n|W=0},P_{Y_n|W=1})$ (see~\cite[Theorem 13.1.1]{LR06}) and the total-variation vanishes
as $n\diverge$.

The slow convergence rates obtained above for Gaussian noise can be explained as follows: In view of \prettyref{eq:TWYn}, the $T$-information obeys the iteration  $T(W; Y_n) \leq \FTV(T(W; Y_{n-1}))$. For instance, consider the Dobrushin curve under unit power constraint is given by $\FTV(t)= t(1 - 2 \sfQ(1/\sqrt{t}))$, which satisfies $\FTV'(0)=1$ and all other derivatives vanish at zero.  Therefore $\FTV$ is smooth but not real analytic at zero, and the rate of convergence of the iteration $x_n  = \FTV(x_{n-1})$ to the fixed point zero is very slow. See \prettyref{fig:tvawgn} for an illustration.

	\label{rmk:rate-awgn}
\end{remark}

\begin{proof}
	
By \prettyref{thm:TVproc}, we have
\begin{equation}
T(W; Y_n) \leq \FTV(T(W; X_n)) \leq \FTV(T(W; Y_{n-1})),	
	\label{eq:TWYn}
\end{equation}
where the first inequality follows from \prettyref{thm:TVproc}, and the second inequality follows from the data
processing theorem for $T$ and the monotonicity of $\FTV$. 
Applying \prettyref{th:coupling}, we have 
\[
		\FTV(t) \leq t \theta_c\left(2 \cost^{-1}\pth{a\over t}\right).
\]
Repeating the above argument leads to 
\[
T(W; Y_n) \leq a t_n,
\]
where the sequence $\{t_n\}$ is defined iteratively via
\begin{equation}
t_{n+1} = t_{n} - h(t_n)
	\label{eq:tn}
\end{equation}
with 
$h(t) = t(1- \theta_c\left(2 \cost^{-1}\pth{\frac{1}{t}}\right))$ and 
$t_1=1$. 
By \prettyref{th:coupling}, $\theta_c$ is strictly increasing. Therefore $h$ is an increasing function.
Applying \prettyref{lmm:rate} in \prettyref{app:rate}, the convergence rate of the sequence \prettyref{eq:tn} satisfies
\begin{equation*}
t_n \leq G^{-1}(n-1) = \frac{1}{f^{-1}(n-1)},	
\end{equation*}
where $G(t) = \int_t^1 \frac{1}{2y (1-\theta_c(2 \cost^{-1}(\frac{1}{y})))}\diff y$.

For the Gaussian noise,  we have $\theta_c(x) =\theta(x) = 1 - 2 \sfQ(x/2)$ (see \prettyref{cor:Ftv-Z}). In view of the bound $\sfQ(u) \geq \frac{\varphi(u) u}{u^2+1}$ for $u>0$, where $\phi$ denote the standard normal density, \prettyref{eq:tvl-g} follows from \prettyref{eq:tvl} upon changes of variables.
\apxonly{YP sanity check: So approximating $\sfQ(\sqrt{x}) \approx {c\over \sqrt{x}} e^{-x^2/2}$ we get
from~\eqref{eq:ptv1} the sequence:
	$$ t_{n+1} = t_n - c\sqrt{t_n} e^{-1/t_n} $$
Approximating by continuous functions we get
$$ {df \over dn} = -c\sqrt{f} e^{-1/f} $$
and solving this for large $n\gg 1$ we obtain
	$$ f(n)^{3/2} e^{1/f(n)} \approx c n $$
	which implies
	$$ f(n) \approx {1\over \ln n} - {3\over 2} {\ln \ln n\over \ln n} + \cdots $$
}
\end{proof}


\subsection{Special case: finite-alphabet $W$}\label{sec:discrete}

A consequence of the total variation estimates in \prettyref{thm:TVproc} and \prettyref{prop:tv-lost} is that for
finitely-valued message $W$ they entail estimates on the mutual information and maximal correlation, as the next
proposition shows.\footnote{The bound~\eqref{eq:ITV1} is essentially~\cite[Lemma 1]{Csiszar96}. The
bound~\eqref{eq:chi_tv} was shown by F. P. Calmon~\texttt{<flavio@mit.edu>} and included here with his permission.}

\begin{proposition}\label{prop:ITV}
Assume $W$ take values on a finite set $\calW$ and let $p_{W, \min}$ denote the minimal non-zero mass of $P_W$.
Then
\begin{align}
I(W;Y)
& \leq  \log (|\calW|-1) T(W;Y) + h(T(W;Y))		\label{eq:ITV1}\\
S^2(W;Y) & \leq \chi^2(P_{WY} \| P_W P_Y)\label{eq:chi_maxcor}\\
&\leq  \frac{1}{p_{W, \min}} T(W;Y) \,,\label{eq:chi_tv}
\end{align}
where $S(W;Y)$ and $\chi^2$ are defined in~\eqref{eq:chi_maxcor0} and~\eqref{eq:chi2def}, respectively, and $h(p)= p \log \frac{1}{p} + (1-p) \log\frac{1}{1-p}$ is the binary entropy function.
\end{proposition}
\begin{proof}
	By coupling and Fano's inequality, for any $P$ and $Q$ on $\calW$, we have
\[
|H(P)-H(Q)| \leq \TV(P,Q) \log (|\calW|-1) + h(\TV(P,Q)).
\]
Then
\begin{align*}
I(W;Y)
= & ~ H(W)-H(W|Y)	\nonumber \\
\leq & ~ 	\Expect_{y \sim P_Y} [\log (|\calW|-1) \TV(P_W, P_{W|Y=y}) + h(\TV(P_W, P_{W|Y=y}))] \\
\leq & ~ 	\log (|\calW|-1) \TV(P_WP_Y, P_{WY}) + h(\TV(P_WP_Y, P_{WY})),
\end{align*}
where the last step is due to the concavity of $h(\cdot)$.

The inequality~\eqref{eq:chi_maxcor} follows~\cite{HW75} by noticing that $\chi^2(P_{WY} \| P_W P_Y)$ is the sum of squares of the singular
values of $f(W) \mapsto \EE[f(W)|Y]$ minus 1 (the largest one), while $S(W;Y)$ is
the second largest singular value.\apxonly{\footnote{Explicit proof of~\eqref{eq:chi_maxcor}: let $Z=f(W)$ with $\EE[Z]=0, \EE[Z^2] = 1$ then 
\begin{align*}
\rho^2
= & ~ \Expect[(\Expect[Z|Y])^2]	=  \int  \frac{\pth{\int z P_{ZY}(z,y)}^2}{P_Y(y)} 
=  \int  \frac{\pth{\int z P_Z(z) (P_{Y|Z}(y|z) - P_Y(y))  }^2}{P_Y(y)} 	\\
\leq & ~ \int  \frac{\Expect[Z^2]  \int (P_{Y|Z}(y|z) - P_Y(y))^2 P_Z(z)}{P_Y(y)} 
= \chi^2(P_{ZY}\| P_Z P_Y) \le \chi^2(P_{WY} \|P_W P_Y) 
\end{align*} }}
Bound~\eqref{eq:chi_tv} follows from the chain:
\begin{align*} 
	\chi^2(P_{WY} \| P_W P_Y) &= \EE_{P_{WY}}\left[{P_{W|Y}(W|Y)\over P_W(W) }\right] - 1\\ 
	&=\EE_{P_{WY}}\left[{P_{W|Y}(W|Y)\over P_W(W) }\right] -\EE_{P_W P_Y}\left[{P_{W|Y}(W|Y)\over P_W }\right] \\
     &\le \esssup_{w,y} {P_{W|Y}(w|y)\over P_W(w)} \cdot \TV(P_{WY}, P_W P_Y)\\
     &\le {1\over p_{W,\min}}  T(W;Y)\,,
\end{align*}     
where first step is by~\eqref{eq:chi2def} and the rest are self-evident.
\end{proof}

Combining Propositions \ref{prop:tv-lost} and \ref{prop:ITV}, we conclude that both $S(W;Y_n)$ and $I(W;Y_n)$ vanish for finitely-valued $W$.
In particular, for Gaussian noise, by \prettyref{rmk:rate-awgn} (second moment constraint) we have $T(W;Y_n) = O(\frac{1}{\log n})$. Then the 
maximal correlation satisfies $S(W;Y_n) = O(\frac{1}{\sqrt{\log n}})$ and the mutual information vanishes according to
\begin{equation}
	I(W;Y_n) = O\pth{h\pth{\frac{1}{\log n}}} = O\pth{\frac{\log \log n}{\log n}}.
	\label{eq:i-rate}
\end{equation}

\apxonly{
\begin{remark}[Maximal correlation and $\chi^2$ divergence]\label{rmk:rho-chi2}
If we had $\chi^2(P_{WY_n}\| P_W P_{Y_n}) \to 0$, then we could immediately derive convergence of maximal correlation and
hence of $\rho^2(W,\Expect[W|Y_n])$ because of the general inequality~\eqref{eq:chi_maxcor}.
Note a curious double-relation between $S(X;Y)$ and $\chi^2$ divergence
in~\eqref{eq:chi_maxcor0} and~\eqref{eq:chi_maxcor}.

Furthermore, $\chi^2$ also would bound the mutual information:
	$$ \chi^2(P_{XY} \| P_X P_Y) \ge \exp\{I(X;Y)\} - 1 $$
(This follows from $\chi^2 = \EE[{P_{XY}\over P_X P_Y}] $ and Jensen.)
\end{remark}

\textbf{Open question:} For any $P_{X_n|X_0}$,
\[
I(X_0;X_n) \to 0 \stackrel{?}{\Longrightarrow} \maxcor(X_0;X_n) \to 0.
\]
Even if $P_{X_n|X_0}$ is cascaded AWGN, it is still unclear. But it is easy to construct examples
where
$$ I(A_n; B_n) \to 0, \maxcor(A_n, B_n) \ge \delta > 0 $$
}

\subsection{Convergence of mutual information}
\label{sec:I}

In this subsection we focus on the AWGN channel and show that the convergence rate \prettyref{eq:i-rate} continues to
hold for \emph{any} random variable $W$, which will be useful for applications in optimal stochastic control where $W$
is Gaussian distributed.\footnote{\textit{Added in print:} Another method of showing~\prettyref{eq:i-rate} is to directly use the strong data processing inequality for mutual information in Gaussian noise, cf.~\cite{yuryITA,CPW15}. Namely, it is possible to show the existence of certain non-linear function $F_I$ such that $F_I(t)<t$ and
\begin{equation}\label{eq:ppdq}
	I(W;X+Z) \le F_I(I(W;X)) 
\end{equation}
for all $(W,X) \dperp Z$ and $\EE[|X|^2]\le E$. Then~\prettyref{eq:i-rate} follows by applying~\prettyref{eq:ppdq}
repeatedly and the behavior of $F_I$ curve near zero: $F_I(t) = t - e^{-\frac{E}{t}\ln \frac{1}{t} + \Theta(\ln \frac{1}{t})}$.}
To deal with non-discrete $W$, a natural idea to apply is \emph{quantization}. By Propositions~\ref{prop:tv-lost} and~\ref{prop:ITV}, for any quantizer $q:\mreals^d\to[m]$,  we have
\begin{equation}\label{eq:tp0}
		I(q(W); Y_n) \le {C \log m \over \log n} + h\left({C \over \log n}\right)
\end{equation}
for some universal constant $C$.
A natural conjecture is the following implication: For any sequence of channels $P_{Y_n|X}$ we
	have:
	$$ \forall m \in \naturals, \forall q: \mreals^d\to[m]: I(q(W); Y_n) \to 0 \quad \implies \quad I(W;Y_n) \to0,
	$$
	which would imply the desired conclusion that mutual information vanishes. 
Somewhat counter-intuitively, this conjecture is generally false, as the following counterexample shows: Consider $X\sim \Unif([0,1])$ and
	$$ Y_n = \begin{cases} 0, {1\over n} \le X \le 1,\\
		k, (k-1){2^{-n}\over n} \le X < k {2^{-n}\over n}, 
		\end{cases} \qquad k=1,\ldots,2^n. $$
	On one hand it is clear that $I(X; Y_n) \to \infty$. On the other hand, 
	among all $m$-point quantizers $q$,
	 it is clear that the
	optimal one is to quantize to some levels corresponding to the partition that $Y_n$
	incurs (other quantizers are just equivalent to randomization). Thus
	$$ \sup_{q:[0,1]\to[m]} I(q(X); Y_n) = \sup_{q:[2^n+1]\to[m]} H(q(Y_n)) .$$
	But the RHS tends to zero as $n\to\infty$ for any fixed $m$ because the dominating
	atom shoots up to $1$. The same example also shows that
\begin{equation}\label{eq:ttoi}
			T(X; Y_n) \to 0 \quad \not\!\!\! \implies \quad I(X;Y_n) \to 0\,. 
\end{equation}	
	\apxonly{Indeed, in this case
		$$ T(X; Y_n) = \EE_{Y_n} \TV(P_{X|Y_n}, P_X) \le \PP[Y_n \neq 0] + O(1/n) = O(1/n)\,.$$
	}%

	Nevertheless, under additional constraints on kernels $P_{Y_n|W}$, we can prove that \prettyref{eq:ttoi} indeed holds and obtain the convergence rate. The main idea is to show that the set of distributions $\{P_{Y_n|W=w}, w\in\mreals^d\}$ can be 	grouped into finitely many clusters, so that the diameter (in KL divergence) of each cluster is arbitrarily small. This can indeed be done in our setting since the channel $P_{Y_n|W}$ is a stochastically degraded version of an AWGN channel.

\begin{proposition}\label{prop:iconv} Let $W, X_k, Y_k$ be as in Theorem~\ref{th:main}. 
If $\Expect[\|X_k\|^2] \leq d E$ for all $k \in [n]$, then
	\begin{equation}\label{eq:tp}
		I(W; Y_n) \le 
\frac{d}{2}\log\pth{1+ \frac{dE}{\log n}} + \frac{d^2E}{2 \log n} \log\pth{1+ \frac{\log n}{d}} + \frac{C d^2E}{\log n} \log\pth{1+\frac{2 \log n}{d \sqrt{E}}} + h\pth{\frac{C dE}{\log n}\wedge1}
,
	\end{equation}
	where $C$ is the absolute constant in \prettyref{eq:main1}.
	In particular, for fixed $d$ and $E$, 
	\begin{equation}
	I(W; Y_n) = O\pth{\frac{\log \log n}{\log n}}.
	\label{eq:tp-d}
\end{equation}
\end{proposition}
\begin{remark} 
Note that the upper bound~\eqref{eq:tp} 
deteriorates as the dimension $d$ grows, which is to be expected. Indeed, for large $d$ one can employ very reliable error-correcting codes for the AWGN channel with blocklength $d$, that can tolerate a large number of hops over the AWGN channels. If the blocklength $d=d_n$ grows with $n$ such that $d_n=O(\log n)$ and the power per coordinate $E$ is fixed, then \prettyref{eq:tp} reduces to
\[
I(W; Y_n) \le O\pth{\frac{d_n^2}{\log n} \log \frac{\log n}{d_n}}.
\]
Using Fano's inequality, this implies that in order to reliably communicate over $n$ hops at some positive rate, thereby $I(W;Y_n)=\Omega(d_n)$, it is necessary to have the blocklength $d_n$ grow at least as fast as 
\begin{equation}
d_n=\Omega(\log n).	
	\label{eq:dn}
\end{equation}
 This conclusion has been obtained in \cite{Subramanian12} under the simplified assumption of almost sure power constraint of the codebook (see \prettyref{eq:mc3sub}). Here \prettyref{prop:iconv} extends it to power constraint in expectation. 
 \apxonly{
 \par To see that \prettyref{eq:dn} is also sufficient for reliable communicate at some positive (albeit small) rate, let $d_n= \rho \log n$ for some $\rho>0$. Note that the optimal code for the AWGN channel with blocklength $d$ and rate $R$ has block error probability at most $\exp(-E(R)d_n) = n^{-\rho E(R)}$, where $E(R)$ is the random coding error exponent. Applying a union bound over the $n$ relays, we conclude that the end-to-end error probability is at most $n^{1-\rho E(R)}$.
Choose $R=R(\rho) >0$, such that $\rho E(R)>1$. }

\label{rmk:rate-pos}
\end{remark}

\begin{proof}[Proof of \prettyref{prop:iconv}]
Fix $u,\epsilon>0$
to be specified later. 
It is well-known that the $\ell_2$-ball in $\reals^d$ of radius $u$ can be covered by at most $m=\floor{(1+\frac{2u}{\epsilon})^d}$ $\ell_2$-balls of radius $\epsilon$, whose centers are denoted by ${x_1},\ldots,{x_m}$. 
Define $q: \reals^d \to [m+1]$ by
\[
q(x) = \bpth{\argmin_{i\in [m]} \|x_i-x\|}\indc{\|x\|\leq u}+ (m+1)\indc{\|x\|> u}.
\]
Then $\expect{\|X_1-x_i\|_2^2|q(X_1)=j} \leq \epsilon^2$ for any $j\in[m]$. Hence
	\begin{align} I(X_1; Y_n | q(X_1) = j) 
	&\le I(X_1; Y_1 | q(X_1) = j)\label{eq:tt1}\\
		&\le \frac{d}{2}\log\pth{1+\frac{\EE[\|X_1 - x_j\|^2 | q(X_1) = j]}{d}}\label{eq:tt2}\\
		&\le \frac{d}{2}\log\pth{1+ \frac{\epsilon^2}{d}}\label{eq:tt3},
	\end{align}
	where in~\eqref{eq:tt1} we used the Markov relation $q(X_1)\to X_1 \to Y_1 \to Y_n$, and \prettyref{eq:tt2} follows from the vector AWGN channel capacity:
	\begin{equation}
	\sup_{P_X: \Expect[\|X\|_2^2 \leq P]} I(X;X+Z) = \frac{d}{2} \log\pth{1+\frac{P}{d}},
	\label{eq:awgn-mi}
\end{equation}
where $Z\sim \calN(0,\mathbf{I}_d)$ is independent of $X$. Similarly,
\begin{align} 
I(X_1; Y_n | q(X_1) = m+1) 
	&\le I(X_1; Y_1 | q(X_1) = m+1) \nonumber \\
		&\le \frac{d}{2}\log\pth{1+\frac{\EE[\|X_1\|^2 | \|X_1\| > u]}{d}} \nonumber\\
		&\le \frac{d}{2}\log\pth{1+ \frac{E}{\prob{\|X_1\| > u}}}\label{eq:tt4},
	\end{align}
where \prettyref{eq:tt4} follows from the fact that $\EE[\|X_1\|^2 | \|X_1\| > u] \prob{\|X_1\| > u} \leq \Expect[\|X_1\|^2]$.

		Averaging \prettyref{eq:tt3} and \prettyref{eq:tt4} over $q(X_1)=j \in [m+1]$, we obtain
		\begin{align}
I(X_1; Y_n | q(X_1)) 
\le & ~  \frac{d}{2}\log\pth{1+ \frac{\epsilon^2}{d}} + \frac{d}{2} \prob{\|X_1\| > u} \log\pth{1+ \frac{E}{\prob{\|X_1\| > u}}}	\nonumber \\
\leq & ~ \frac{d}{2}\log\pth{1+ \frac{\epsilon^2}{d}} + \frac{d^2E}{2u^2} \log\pth{1+ \frac{u^2}{d}},
				\label{eq:tt0x}
\end{align}
where \prettyref{eq:tt0x} follows from the fact that $x \mapsto x \ln(1+\frac{1}{x})$ is increasing on $\reals_+$\footnote{Indeed, $(x \ln(1+\frac{1}{x}))'=-\ln(1-\frac{1}{1+x})-\frac{1}{1+x} \geq 0$.} and the Chebyshev's inequality:
\begin{equation}
	\prob{\|X_1\| \geq u} \leq \frac{\Expect[\|X_1\|^2]}{u^2} \leq \frac{dE}{u^2}.
	\label{eq:cheb}
\end{equation}
Applying \prettyref{prop:ITV}, we have
\begin{equation}
	I(q(X_1); Y_n) \leq t_n d \log \pth{1+\frac{2u}{\epsilon}} + h(t_n),
	\label{eq:tt5}
\end{equation}
where $t_n = T(q(X_1);Y_n) \leq T(X_1;Y_n) \leq \frac{C dE}{\log n}$ in view of \prettyref{eq:main1}

Combining \prettyref{eq:tt4} and \prettyref{eq:tt5} yields
	\begin{align} 
	I(W;Y_n)
	&\leq I(X_1; Y_n) \\
	&= I(q(X_1); Y_n) + I(X_1; Y_n | q(X_1))\\
		&\le \frac{d}{2}\log\pth{1+ \frac{\epsilon^2}{d}} + \frac{d^2E}{2u^2} \log\pth{1+ \frac{u^2}{d}} + \frac{C d^2E}{\log n} \log\pth{1+\frac{2u}{\epsilon}} + h\pth{\frac{C dE}{\log n}\wedge1}.\label{eq:tt6}
	\end{align}
Choosing $u=\sqrt{\log n}$ and $\epsilon^2 = \frac{d^2 E}{\log n}$ yields the desired \prettyref{eq:tp}.
\end{proof}

\subsection{Convergence of correlation coefficients}
\label{sec:rho}
Given a pair of random variables $X,Y$, the conditional expectation of $X$ given $Y$ has the maximal correlation with $X$ among all functions of $Y$, \ie
\[
\sup_{g \in L_2(P_Y)} \rho(X,g(Y)) = \rho(X,\Expect[X|Y]) = \frac{\|\Expect[X|Y]-\Expect[X]\|_2}{\sqrt{\var(X)}}\,,
\]
which is a simple consequence of the Cauchy-Schwartz inequality.
As the next result shows, vanishing mutual information provides a convenient sufficient condition for establishing vanishing correlation coefficients.
\begin{proposition} \label{prop:rconv}
Assume that $\Expect [W^2] < \infty$. For any sequence of $P_{Y_n|W}$, 
\begin{equation}
\lim_{n\to\infty} I(W; Y_n) = 0 \quad \implies \quad \lim_{n\to\infty} \rho(W,  \Expect[W|Y_n]) = 0.	
	\label{eq:Irho}
\end{equation}
Moreover, if $W$ is Gaussian, then
\begin{equation}
\rho^2(W, \Expect[W|Y_n]) \leq 1 - \exp(- 2 I(W; Y_n)) \le 2 I(W; Y_n).
	\label{eq:Irho-g}
\end{equation}
\end{proposition}
\begin{proof}
For the Gaussian case, \prettyref{eq:Irho-g} follows from the inequality
\begin{equation}
I(W;\hW) \geq \frac{1}{2} \log \frac{1}{1-\rho^2(W,\hW)},	
	\label{eq:Icor-g}
\end{equation}
which is equivalent to the Gaussian rate-distortion formula.
To see the implication \prettyref{eq:Irho}, first notice the equivalence
$$  \EE[(W-\Expect[W|Y_n])^2] \to \var(W) \quad\iff\quad
\rho(W, \Expect[W|Y_n]) \to 0.$$
From here Proposition~\ref{prop:rconv} follows from the next (probably well-known) lemma.
\end{proof}
		
\begin{lemma}
Assume that $\Expect[X^2]  < \infty$. Let $\var(X) = \sigma^2$. Denote the rate-distortion function of $X$ with respect to the mean-square error by
\[
R(D) = \inf_{P_{\hX|X}: \Expect (\hX-X)^2 \leq D} I(X; \hX).
\]
Then
\begin{equation}
	D \to \sigma^2 \Leftrightarrow R(D) \to 0.
	\label{eq:RD}
\end{equation}
	\label{lmm:RD}
\end{lemma}
\begin{proof}
($\Rightarrow$)  The rate-distortion function is dominated by that of the Gaussian distribution \cite{berger}: 
\begin{equation}
R(D) \leq \frac{1}{2} \log^+ \frac{\sigma^2}{D},	
	\label{eq:RG}
\end{equation}
 where $\log^+ \triangleq \max\{\log, 0\}$.

($\Leftarrow$)  Note that $D \mapsto R(D)$ is decreasing and concave on $[0,\sigma^2]$, hence continuous on the open
interval $(0,\sigma^2)$. Suppose there exists $D_0 < \sigma^2$ such that $R(D_0)=0$. Then by definition of the
rate-distortion function, there exists a sequence of $P_{\hX_n|X}$ such that $\Expect(\hX_n-X)^2\leq D$.  $I(X;\hX_n)
\to 0$. Note that $\Expect X_n^2 \leq 2 D + 2\Expect X^2$ for all $n$. Therefore the sequence $P_{\hX_n,X}$ is tight. By
Prokhorov's theorem, there exists a subsequence $P_{\hX_{n_k},X}$ which converges weakly to some $P_{\hX,X}$. 
By the lower semicontinuity of the divergence and the second-order moment, $\Expect(\hX_n-X)^2 \leq \liminf
\Expect(\hX_{n_k}-X)^2\leq D$ and $I(\hX; X) \leq \liminf I(\hX_{n_k};X) = 0$. Hence $\hX \indep X$, contradicting
$\Expect(\hX_n-X)^2 \leq D < \sigma^2$. 
\end{proof}

\smallskip

\prettyref{prop:rconv} allows us to capitalize on the results on mutual information in \prettyref{sec:I} to obtain
correlation estimates for the Markov chain \prettyref{eq:mc1}. In particular, combining \prettyref{eq:Irho} with
\prettyref{prop:iconv} yields \prettyref{eq:main2}. Additionally, if $W$ is Gaussian, then \prettyref{eq:Irho-g} yields 
\begin{equation}
\rho(W,\Expect[W|Y_n]) = O(\sqrt{I(W; Y_n)}) = O\pth{\sqrt{\frac{\log \log n}{\log n}}}.	
	\label{eq:main3-rate}
\end{equation}
These prove the correlation part of the main result \prettyref{th:main}.

However, the estimate \prettyref{eq:main3-rate} is not entirely satisfactory in the sense that it highly depends on the Gaussianity of $W$; 
if $W$ is not Gaussian, the rate-distortion function of $W$ is not explicitly known and it is unclear whether \prettyref{eq:Irho-g} still applies. 
How to obtain quantitative estimates on the correlation coefficient if we only have sub-Gaussianity or moment constraints on $W$? It turns out that one can circumvent mutual information completely and directly obtain correlation estimate from the $T$-information, whose convergence rate has been found in \prettyref{sec:T}. The key connection between total variation and correlation is the following simple observation:

\begin{proposition}
\label{prop:rconv-ng}
Assume $W$ is zero-mean, unit-variance. For any $q \in (1,\infty]$ we have
\begin{equation}
	\rho^2(W,\Expect[W|Y]) \leq 4 T(W;Y)^{1-\frac{1}{q}} \|W\|_{2q}^2.
	\label{eq:rconv-moment}
\end{equation}
If $W$ is sub-Gaussian and $T(W;Y) < e^{-2/e}$, we have
\begin{equation}
	\rho^2(W,\Expect[W|Y]) \leq {8\over \log e} \|W\|^2_{\psi_2} \, T(W;Y) \log \frac{1}{T(W;Y)} \,.
	\label{eq:rconv-subg}
\end{equation}
where $\|W\|_{\psi_2} \triangleq \inf\{c>0: \Expect [e^{W^2/c^2}] \leq 2\}$ is an Orlicz norm.
\end{proposition}

\prettyref{prop:rconv-ng} is reminiscent of Tao's inequality \cite{tao.szemeredi, ahlswede.tao} and \cite[Theorem 10]{mmse.functional.IT},  which use mutual information to produce correlation estimates for bounded random variables: If $\var(W)=1$, then
\[
\rho^2(W,\Expect[W|Y]) \leq \frac{2}{\log \eexp} \|W\|_\infty^2 I(W;Y).
\]
In contrast, \prettyref{prop:rconv-ng} uses $T$-information in lieu of mutual information and allows more general tail condition.
\apxonly{Quick proof:
\[
|\Expect[W]-\Expect[W|Y]| \leq \int |w| |P_W-P_{W|Y}| \leq 2 \|W\|_\infty \TV(P_W,P_{W|Y}).
\]
Then $\|\Expect[W]-\Expect[W|Y]\|_2^2 \leq 4 \|W\|_\infty^2 \Expect[\TV(P_W,P_{W|Y})^2] \leq 2 \|W\|_\infty^2 I(W;Y)$. Note that if we use $\TV \leq 1$, then we get 
$\|\Expect[W]-\Expect[W|Y]\|_2^2 \leq 4 \|W\|_\infty^2 T(W;Y)$, which is exactly \prettyref{eq:rconv-moment} with $q=\infty$.
}

\begin{remark}
Combining \prettyref{prop:rconv-ng} with the convergence rate of the $T$-information in \prettyref{prop:tv-lost}, we obtain the corresponding convergence rate of correlation under various cost constraints on the relays and tail conditions on the original message $W$. For example, in view of \prettyref{rmk:rate-awgn},	if the cost function is $\cost(x)=|x|^p$ and $W$ is sub-Gaussian, then
\begin{equation}
\rho(W,\Expect[W|Y_n]) = O\pth{\frac{\sqrt{\log \log n}}{(\log n)^{p/4}}}.	
	\label{eq:corr-pmoment-subg}
\end{equation}
In particular, for average power constraint ($p=2$), the convergence rate \prettyref{eq:main3-rate} applies to all sub-Gaussian $W$. We will show in the next subsection that \prettyref{eq:corr-pmoment-subg} is in fact optimal for all $p$ when $W$ is Gaussian. 
\apxonly{Additionally, if $W$ has $q\Th$ moment for $q>2$, then $\rho(W,\Expect[W|Y]) = O((\log n)^{p/q-p/2})$.}
	\label{rmk:rcov-ng}
\end{remark}

\begin{proof}[Proof of \prettyref{prop:rconv-ng}]
 Since $T(W; Y)=\EE[\TV(P_{W|Y}, P_W)]$ we may construct a probability space with three variables $W,W',Y$
	such that $W'\dperp Y$ and furthermore
		$$ \PP[W\neq W'] = T(W; Y). $$
		Then, consider an arbitrary zero-mean $g(Y)$ and write
		\begin{align} \EE[W g(Y)] = \EE[W g(Y)] - \EE[W' g(Y)] &\le \EE[|g(Y)| \cdot |W-W'| \indc{W\neq W'}]
			\label{eq:rcm1a}\\
			&\le \|g(Y)\|_2 \|W-W'\|_{2q} T(W; Y)^{1\over 2q'}\,,\label{eq:rcm1}
	\end{align}		
	where the last step is by H\"older's inequality since ${1\over 2} + {1\over 2q} + {1\over 2q'} =1$ and 
	$q'={q\over q-1}$.
	 Since $\|W-W'\|_{2q} \le 2\|W\|_{2q}$, normalizing both sides of~\eqref{eq:rcm1} by $\|g(Y)\|_2$ and $\|W\|_2=1$ yields the desired \prettyref{eq:rconv-moment}.

	 For the second part of the proposition, consider arbitrary non-negative, convex $\psi:\mreals\to\mreals$ with
	 $\psi(0)<1$ and define the following Orlicz norm
	 	$$ \|X\|_\psi \eqdef \inf\{c>0: \EE[\psi(X/c)] \le 1\}\,.$$
	 If $\psi^*$ is the Legendre dual of $\psi$ then from Young's inequality we have for arbitrary $X,Y$:
	 $$ XY \le \psi(X) + \psi^*(Y), $$
	 and, hence,
	 \begin{equation}\label{eq:young_orlicz}
		 	\EE[XY] \le 2 \|X\|_{\psi} \|Y\|_{\psi^*} 
  	\end{equation}	 
	 Consider 
	 \begin{align*} \psi_1(x) &= {1\over 2} e^x, \qquad \psi_1^*(y)=y \ln {2y\over e},\\
		   \psi_2(y) &= {1\over 2} e^{x^2} ,
	   \end{align*}		   
	  and notice an easy identity
	  \begin{equation}\label{eq:rcm2}
		  	\|X^2\|_{\psi_1} = \|X\|_{\psi_2}^2 .
  \end{equation}	  
  Then, proceeding as above we only need to upper-bound $\EE[|W-W'|^2 \indc{W\neq W'}]$ in~\eqref{eq:rcm1a}. From
	  inequality~\eqref{eq:young_orlicz} and~\eqref{eq:rcm2} we get
	  $$ \EE[|W-W'|^2 \indc{W\neq W'}] \le 2 \|W-W'\|_{\psi_2}^2 \|\indc{W\neq W'}\|_{\psi_1^*}\,.$$
	  For the first term we apply triangle inequality.
	  The $\psi_1^*$-norm of the indicator is found as a unique solution of
	  	$$ c = \PP[W\neq W'] \ln{2\over ec}\,,$$
		with $c\in(0, {2\over e})$. It is easy to show that if $\PP[W\neq W'] <e^{-2/e}$ then 
	  $$ \|\indc{W\neq W'}\|_{\psi_1^*} \le \PP[W\neq W'] \ln {1\over \PP[W\neq W']}\,,$$
	  from which the proposition follows.
\end{proof}

\apxonly{
\begin{remark}
YW's peasant proofs of \prettyref{prop:rconv-ng}: Part 1  based on Hellinger:
\begin{align*}
\rho^2
= & ~ \Expect[(\Expect[W|Y])^2]	\nonumber \\
= & ~ \int P_Y(\diff y) \pth{\int w (P_W(\diff w) - P_{W|Y=y}(\diff w) ) }^2	\\
\leq & ~ \int P_Y(\diff y) \pth{\int w^2 (\sqrt{P_W(\diff w)} + \sqrt{P_{W|Y=y}(\diff w)})^2 } \pth{\int (\sqrt{P_W(\diff w)} - \sqrt{P_{W|Y=y}(\diff w)})^2 }	\\
\leq & ~ \int P_Y(\diff y) \pth{\int 2 w^2 (P_W(\diff w)+ P_{W|Y=y}(\diff w)) } H^2(P_W, P_{W|Y=y})	\\
\leq & ~ \int P_Y(\diff y) \pth{\int 2 w^2 (P_W(\diff w)+ P_{W|Y=y}(\diff w)) } 2 \TV(P_W, P_{W|Y=y})	\\
\leq & ~ 2 T(W;Y) + 2 \Expect[\TV(P_W, P_{W|Y}) \Expect[W^2|Y] ] \\
\leq & ~ 2 T(W;Y) + 2 T(W;Y)^{1-1/q} \|\Expect[W^2|Y]\|_q \\
\leq & ~ 2 T(W;Y) + 2 T(W;Y)^{1-1/q} \|W\|_{2q}^2 
\end{align*}
\end{remark}

Part 2 based on relation between moments and Orlicz norms:
If in addition that $W$ is sub-Gaussian, then we have $\|W\|_p \leq c \|W\|_{\psi_2} \sqrt{p}$ for any $p \geq 1$ and some universal constant $c>0$ (see, \eg, \cite[Lemma 5.5]{Vershynin10}). 
Optimizing \prettyref{eq:rconv-moment} over $q>1$, \ie, choosing $q = \log \frac{1}{T(W;Y)}$, we obtain
\prettyref{eq:rconv-subg}. Note: sub-Gaussian moment bound for even $p=2q$: $x^{2q} \leq \frac{1}{q!} \exp(x^2)$
then apply to $x=W/\|W\|_{\psi_2}$.
}

\subsection{Achievable schemes}
	\label{sec:ach}

For the scalar case we construct a relay scheme under which the $T$-information, mutual information and the correlation  
between the initial message $W\sim\matn(0,1)$ and the final output $Y_n$ achieve the lower bounds \prettyref{eq:main1} -- \prettyref{eq:main2} up to constants. 
This scheme is also useful for the optimal control problem in \prettyref{sec:control}.
For simplicity we only consider the $p\Th$ moment constraint $\Expect|X_k|^p \leq a$ and assume $W\sim \calN(0,1)$ and $a=2$ for notational conciseness.

\paragraph{Binary-messaging scheme}
In view of the converse results in Sections \ref{sec:T} -- \ref{sec:rho}, the majority of the information will be inevitably lost regardless of the relay design. Thus we only aim to transmit a small fraction of the original message, \eg, a highly skewed quantized version, reliably. To this end, let 
\begin{equation}
\mu = 4\sqrt{\log n}, \quad a = \sfQ^{-1}(\mu^{-p}) = \sqrt{p\log \log n} +o(1).	
	\label{eq:mua}
\end{equation}
 Let $X_1 = \mu \indc{W \geq a}$, which satisfies $\Expect |X_1|^p = 1$.
At each stage, the relay decodes the previous message by $X_{k+1}=\mu \indc{Y_k \geq \mu/2}$. 
Note that all $X_k$'s take values in $\{0,\mu\}$.
Then $\prob{X_{k+1}\neq X_k} \leq \prob{|Z_k| \geq \mu/2} = 2 \sfQ(\mu/2)$. For any $k\in[n+1]$, applying the union bound and the fact that 
 $\sfQ(a) \leq \varphi(a)/a$, we obtain
\begin{equation}
\prob{X_k \neq X_1} \leq 2 n \sfQ(\mu/2) \leq n^{-1}	.
	\label{eq:XkX1}
\end{equation}
Moreover, the moment constraint is satisfied since
$$\Expect |X_k|^p = \mu^p \prob{X_k \neq 0} \leq \mu^p (\prob{X_1 \neq 0}+\prob{X_k \neq X_1}) \leq 1+ \frac{1}{n}(16\log n)^{p/2} \leq 2$$ for all sufficiently large $n$.

\paragraph{Total variation and Mutual information}
We show that 
\begin{align}
T(W;Y_n) &~= \Omega\pth{\frac{1}{(\log n)^{p/2}}},	\label{eq:tvt0} \\
I(W;Y_n) &~= \Omega\pth{\frac{\log \log n}{(\log n)^{p/2}}},	\label{eq:mi-ach}
\end{align}
which matches the upper bound in \prettyref{rmk:rate-awgn} and the upper bound \prettyref{eq:tp-d} (for $p=2$), respectively.
 Since $X_1$ and $X_{n+1}$ are deterministic functions of $W$ and $Y_n$, respectively, we have $X_1 \to W \to Y_n \to X_{n+1}$ and
\[
T(W;Y_n) \geq T(X_1;X_{n+1}) \geq \mu^{-p} - \prob{X_{n+1} \neq X_1} = \Omega((\log n)^{-p/2}),
\]
where the first inequality follows from data processing, the second inequality follows from \prettyref{eq:T-binary}, and the last inequality is by \prettyref{eq:XkX1}.
Similarly, 
\[
I(W;Y_n) \geq I(X_1;X_{n+1}) = H(X_1)-H(X_1|X_{n+1}) \geq h(\mu^{-p/2}) - h(1/n) = \Omega(\mu^{-p/2} \log \mu ).
\]

\paragraph{Correlation}
Denote $B = \indc{W \geq a} = X_1/\mu$ and $\hat W = \Expect[W | B]=g(B)$, where 
\begin{equation}
	g(0) = \Expect[W|W\leq a] = -\frac{\varphi(a)}{\Phi(a)}, \quad g(1) = \Expect[W|W>a] = \frac{\varphi(a)}{\sfQ(a)}.
	\label{eq:g}
\end{equation}
Using the fact that $\sfQ(x) = \frac{\varphi(x)}{x}(1+o(1))$ as $x\diverge$, we have
\begin{equation}
\Expect[W\hW]=\Expect[\hW^2]=\frac{\varphi^2(a)}{\Phi(a) \sfQ(a)} = \sfQ(a) a^2 (1+o(1)) = \Theta\pth{\frac{\log \log n}{(\log n)^{p/2}}},
	\label{eq:hWnorm}
\end{equation}
where the last inequality follows from the choice of $a$ in \prettyref{eq:mua}.

Set $B_n = \indc{Y_n\geq \mu/2} = X_{n+1}/\mu$ and $W_n = g(B_n)$. By \prettyref{eq:XkX1}, we have $\Prob[B \neq B_n] \leq \frac{1}{n}$.
Therefore 
\begin{align}
\Expect[WW_n]
= & ~ \Expect[W\hW] + \Expect[W(W_n-\hW) \indc{B \neq B_n}]  	 \nonumber \\
\geq & ~ \Expect[W\hW] - \max\{|g(0)|,|g(1)|\} \Expect[|W| \indc{B \neq B_n}]  \nonumber \\
\geq & ~ \Expect[W\hW] - g(1) \sqrt{\Prob[B \neq B_n]} \label{eq:cor1} \\
= & ~ \Theta\pth{\frac{\log \log n}{(\log n)^{p/2}}}, \label{eq:cor2}
\end{align}
where \prettyref{eq:cor1} is by Cauchy-Schwartz, 
\prettyref{eq:cor2} is by \prettyref{eq:hWnorm} and $g(1) = a(1+o(1)) = \Theta(\sqrt{\log \log n}) $. Similarly,
\begin{align}
|\Expect[W_n^2] - \Expect[\hW^2]|
= & ~  |\Expect[(W_n^2-\hW^2) \indc{B \neq B_n}]| \leq g(1)^2 \Prob[B \neq B_n] = O(\log \log n/n).
	\label{eq:hV-norm}
\end{align}
Therefore $\|W_n\|_2=\|\hat W\|_2(1+o(1))$.
Consequently, the correlation satisfies
\begin{equation}
\rho(W,\Expect[W|Y_n]) = \sup_{g\in L_2(P_{Y_n})} \rho(W, g(Y_n)) \geq \rho(W, W_n) = \frac{\Expect[WW_n]}{\|W_n\|_2} =
\Omega\pth{\frac{\sqrt{\log\log n}}{(\log n)^{p/4}}},	
	\label{eq:rho-opt}
\end{equation}
which meets the upper bound \prettyref{eq:corr-pmoment-subg}.

\section{Applications}\label{sec:appl}

\subsection{Optimal memoryless control in Gaussian noise}
	\label{sec:control}
	
	The problem of optimal memoryless control in Gaussian noise was investigated in	\cite{LM11}. Consider the $n$-stage stochastic control problem in \prettyref{fig:control} in one dimension ($d=1$) where the input $W=X_0+Z_0$ with $X_0\sim\calN(0,\sigma_0^2)$ independent of $Z_0\sim\calN(0,1)$. The additive noise $Z_1,\ldots,Z_n$ are \iid standard Gaussian, and the relay function $f_j$ plays the role of a memoryless controller mapping the noisy observation $Y_{i-1}$ into a control signal $X_i$. 
	Let $X_{n+1}=f_{n+1}(Y_n)$ denote the final estimate. Then we have the following Markov chain which has two more stages than \prettyref{eq:mc1}:
	\[
	X_0 \to	W \to X_1 \to Y_1 \to X_2 \to Y_2 \to \cdots \to X_n \to Y_n \to X_{n+1}.
	\]
The major difference is that,	instead of requiring that each controller satisfies the same power constraint as in \prettyref{eq:mc3}, here only a \emph{total} power budget is imposed:
\begin{equation}
\sum_{j=1}^n \EE[X_j^2] \le nE\,.
	\label{eq:total-power}
\end{equation} 
The objective is to maximize the correlation between $X_0$ and $X_{n+1}$.

The main results of \cite{LM11} show that although linear controllers are optimal for two stages ($n=1$) \cite[Proposition 7]{LM11}, for multiple stages they can be strictly sub-optimal. Specifically, subject to the constraint \prettyref{eq:total-power}, the optimal squared  correlation $\rho^2(X_0,X_{n+1})$
achieved by linear controllers is \cite[Lemma 6]{LM11} 
\begin{equation}
\frac{\sigma_0^2}{1+\sigma_0^2} \pth{\frac{E}{1+E}}^{n} \,,
	\label{eq:LM11-affine}
\end{equation}
which vanishes exponentially as $n \diverge$. \cite[Theorem 15]{LM11} shows that \prettyref{eq:LM11-affine} can be improved by using binary quantizers in certain regimes, although the correlation still vanishes exponentially fast albeit with a better exponent. The optimal performance of non-linear controllers is left open in \cite{LM11}.
	
Capitalizing on the results developed in \prettyref{sec:proof}, next we show that 
	the squared correlation achieved by the best non-linear controllers is $\Theta(\frac{\log \log n}{\log n})$, which is significantly better than the exponentially small correlation \prettyref{eq:LM11-affine} achieved by the best linear controllers. 
	\begin{itemize}
	\item For any sequence $\{f_j\}$ satisfying the total power constraint \prettyref{eq:total-power}, the correlation necessarily satisfies 
	\begin{equation}
\rho^2(X_0,X_{n+1}) = O\left(\frac{\log \log n}{\log n}\right).	
	\label{eq:control-optimal}
\end{equation}
 To see this, applying the data processing inequality \prettyref{thm:TVproc} and the $\FTV$ curve in \prettyref{cor:Ftv-g} with $\cost(|x|)=|x|^2$, we have 
 \[
 T(X_0;X_{n+1}) \leq T(W;Y_n) \leq F_1 \circ \cdots \circ F_n (1),
 \]
  where $F_i(t) = t(1-2\sfQ(\sqrt{a_i/t}))$ and $a_i=\Expect[X_i^2]$. Since $\sum_{i=1}^n a_i \leq n E$, we have $\sum_{i=1}^n \indc{a_i \geq 2E} \leq n/2$. Consequently, \prettyref{prop:tv-lost} applies with $n$ replaced by $n/2$ and, by \prettyref{rmk:rate-awgn}, we have $T(W;Y_n) \leq C/\log n$ for some constant $C$ only depending on $E$. Since $X_0$ is Gaussian, applying \prettyref{prop:rconv-ng} yields the upper bound \prettyref{eq:control-optimal}.
  
	\item Conversely, the binary-quantizer scheme described in \prettyref{sec:ach} (with $p=2$) achieves 
	\[
	\rho^2(X_0,X_{n+1})= \Omega\pth{\frac{\log \log n}{\log n}}.
	\]
	Set $X_{n+1} = W_n= g(\indc{Y_n\geq \mu/2})$, where $g$ is defined in \prettyref{eq:g}. Since $W=X_0+Z_0$ and $V \triangleq X_0-\sigma_0^2 Z_0$ are independent, we have $\Expect[X_0X_{n+1}] = \frac{\sigma_0^2}{1+\sigma_0^2} \Expect[WW_n]$ and the rest follows from \prettyref{eq:rho-opt}.
\end{itemize}

The fact that linear control only achieves exponentially decaying correlation can also be understood from the perspective of contraction coefficient of KL divergence. Note that if all controllers are linear, then all input $X_i$'s to the AWGN channel are Gaussian. Recall the distribution-dependent contraction coefficient $\etaKL(Q)$ defined in \prettyref{eq:eta_fq}. For AWGN channel with noise variance $\sigma^2$ and Gaussian input with variance $P$, Erkip and Cover showed in \cite[Theorem 7]{EC98} that
 $\etaKL(\calN(\mu,P)) = \frac{P}{P+\sigma^2}$, which is strictly less than one. This results in exponentially small mutual information:
 \begin{align*}
 I(W;\hW) 
\leq & ~ I(W;Y_1) \prod_{i=2}^n \etaKL(\calN(\Expect[X_i],\var(X_i)))	\nonumber \\
\leq & ~ \frac{\log(1+\sigma_0^2)}{2} \prod_{i=2}^n \frac{\var(X_i)}{1+\var(X_i)} \leq \frac{\log(1+\sigma_0^2)}{2} \pth{\frac{E}{1+E}}^{n-1},
\end{align*}
 where the last step follows from \prettyref{eq:total-power} and the concavity and monotonicity of $x \mapsto \log \frac{x}{x+1}$. Together with the Gaussian rate-distortion function \prettyref{eq:Icor-g}, this implies $\rho(W,\hW)$ must vanish as $(\frac{E}{1+E})^n$ which agrees with \prettyref{eq:LM11-affine}.
 Therefore from a control-theoretic perspective, it is advantageous to design the controller to steer the output away from Gaussian, which requires, of course, non-linear control.

\subsection{Uniqueness of Gibbs measures}
\label{sec:gibbs}

In this section we rely on the notations and results from the theory of infinite-volume Gibbs measures; in particular we
assume familiarity with~\cite[Chapter 2]{HOG11}. Consider a $\mreals$-valued Markov random field $\{X_n: n\in\integers\}$ specified by
pairwise potentials $\Phi_j(x_j, x_{j+1})$. We assume that for every $k\in\integers$ and every $L\ge 1$ we have
$$ \int \exp\left\{- \sum_{j=k}^{k+L} \Phi_j(x_j, x_{j+1})\right\} dx_k \cdots dx_{k+L} < \infty\,. $$
This specification translates into requiring the conditional probabilities to be of the following form:
\begin{equation}\label{eq:gibbs1}
		P_{X_{k+1}^{k+L} | X_{-\infty}^k, X_{k+L+1}^{\infty}} \propto 
	\exp\left\{- \sum_{j=k}^{k+L} \Phi_j(x_j, x_{j+1})\right\} dx_k \cdots dx_{k+L}\,,
\end{equation}	
and in particular $X_n$ form a doubly-infinite Markov chain:
\begin{equation}\label{eq:gibbs2}
		\cdots - X_{-1} - X_0 - X_1 - \cdots 
\end{equation}	

One of the principal questions in Gibbs theory is: Do there exist
none, one or many joint distributions satisfying conditional probabilities~\eqref{eq:gibbs1}? 
Such a joint distribution is called a Gibbs measure consistent with the specification \prettyref{eq:gibbs1}.
It is believed that
the existence of multiple Gibbs measures corresponds to the existence of second-order phase transitions in physics (such as the Curie
temperature in ferromagnets).

A typical method for proving non-existence of multiple phases is the application of Dobrushin contraction, cf.~\cite{RLD70}.
Next we extend this technique to cases where Dobrushin contraction is not available
($\etaTV = 1$) by relying on the knowledge of the Dobrushin curve $\FTV$. Here is an illustration.

\begin{theorem} Suppose that potentials $\Phi_j$ are such that each conditional distribution~\eqref{eq:gibbs1} factors
	through the Gaussian channel, i.e. for each $k, L$ there exists a representation
	\begin{equation}\label{eq:gibbs4}
			P_{X_{k+1}^{k+L} | X_{-\infty}^k, X_{k+L+1}^{\infty}} = P_{X_{k+1}^{k+L}|Y}\circ P_{Y|X_k, X_{k+L+1}}\,,
\end{equation}	
	with $P_{Y|X_k, X_{k+L+1}}$ a two-dimensional Gaussian channel~\eqref{eq:mc2}. Then there may exist at most one
	joint distribution of $X_{-\infty}^{\infty}$ satisfying
	\begin{equation}\label{eq:gibbs3}
			\sup_{j\in\integers} \EE[|X_j|^2]  < \infty .
\end{equation}
\label{thm:gibbs}	
\end{theorem}

\begin{remark} Assumptions of \prettyref{thm:gibbs} guarantee that ``strengths'' of all links in~\eqref{eq:gibbs2} are uniformly
	upper-bounded. Thus we can see that on $\integers$ the only possibilities for a phase transition are: 1) when
	the links become asymptotically noiseless, or 2) when the (non shift-invariant) solutions are allowed to grow
	unbounded. This is in accord with known examples of systems with non-unique Gibbs measures: e.g., the
	asymptotically noiseless example in~\cite[Chapter 6]{HOG11}, or the non shift-invariant examples of 
	Spitzer-Cox and Kalikow in~\cite[Chapter 11]{HOG11}.\apxonly{\\ \textbf{TODO:} For
	the revision try to insert some simple sufficient condition that would imply~\eqref{eq:gibbs4}. Something like a
	uniform quadratic lower-bound $\Phi_j(x_j, x_{j-1}) \ge \epsilon (x_j-x_{j-1})^2$.}
\end{remark}

\begin{proof}

We recall the following idea due to Dobrushin \cite[Lemma 5]{RLD70}:
\begin{proposition}\label{pr:dobr} Let $\pi$ be any coupling of $P_{AB}$ to $Q_{AB}$ (i.e. $\pi_{ABA'B'}$ is $P_{AB}$ or $Q_{AB}$ when
	restricted to first pair or second pair).
		Assume also that for every $a$ and $a'$ we have\footnote{Here $W_\rho$ is a Wasserstein
		distance with respect to the metric $\rho$, analogously defined as in \prettyref{eq:w1} with the $L_1$ distance replaced by $\rho$.}
		$$ W_\rho(P_{B|A=a}, Q_{B|A=a'}) \le r(a,a') \,.$$
	Then there exists a coupling $\tilde \pi$ between $P_{AB}$ and $Q_{AB}$ such that $\tilde \pi_{A,A'} = \pi_{A,A'}$ and 
	$$ \EE_{\tilde \pi}[\rho(B,B')] \le \EE_\pi[r(A,A')]. $$
\end{proposition}
\apxonly{Proof: Just take
$$ \tilde \pi_{B, B'|A=a, A'=a} = \mbox{best coupling of~} P_{B|A=a} \mbox{~to~} Q_{B|A=a'} $$}

	When $r(a,a') = c \rho(a,a')$ and $c<1$ (Dobrushin contraction), we can progressively refine the coupling at
	various points between two distributions $P$ and $Q$ and show that they must coincide. This is a brilliant idea
	of Dobrushin~\cite{RLD70}. We apply the same recursion here, except without relying on $c<1$.

	Suppose that there exist two distributions $P$ and $Q$ of $X_{-\infty}^\infty$ satisfying~\eqref{eq:gibbs4} and \eqref{eq:gibbs3}. 
	Let $E>0$ denote the left-hand side of \prettyref{eq:gibbs3}, \ie, 
	the common upper bound on the second moment of $X_j$. Given a coupling $\pi$ between $P$ and $Q$, that is
		$$ \pi_{X_{k}^n} = P_{X_k^n}, \quad \pi_{\tilde X_{k}^n} = Q_{X_k^n}, \quad k \leq n$$
	denote
	$$ \epsilon_N = \pi[X_{-N}^N \neq \tilde X_{-N}^N] \le 1, $$
	where $N\ge 1$ is large integer.
	
	Denote $x_{\pm N} = (x_N,x_{-N})$ and $|x_{\pm N}| = \sqrt{x_N^2+x_{-N}^2}$ its Euclidean norm.
	Using the factorization condition~\eqref{eq:gibbs4} and the data processing inequality for total variation, we have
	\begin{align*}
	& ~ \TV\pth{P_{X_{-N+1}^{N-1}|X_{\pm N}=a_{\pm N}}, Q_{X_{-N+1}^{N-1}|X_{\pm N}=b_{\pm N}}} \\
\leq & ~ \TV(P_{Y|X_{\pm N}=a_{\pm N}}, Q_{Y|X_{\pm N}=b_{\pm N}}) = \TV(\calN(a_{\pm N},\vect I_2),\calN(b_{\pm N},\vect I_2))	\nonumber \\
= & ~ \theta_c\left(|a_{\pm N}-b_{\pm N}|\right)\,,
\end{align*}
	where $\theta_c(u) = 1-2\sfQ(u/2)$, cf. Corollary~\ref{cor:Ftv-g}. Applying
	Proposition~\ref{pr:dobr} with $\rho(a,a')=\indc{a\neq a'}$ and $r(a,a')=\theta_c(|a-a'|)$, we can produce a new coupling $\pi'$ 
	so that $\pi_{X_{\pm N}, \tX_{\pm N}}' = \pi_{X_{\pm N}, \tX_{\pm N}}$ and 
	$$ \pi'[X_{-N+1}^{N-1} \neq \tilde X_{-N+1}^{N-1}] \le \EE_\pi[\theta_c(|X_{\pm N} - \tilde X_{\pm N}|) ] .$$
	In view of the moment constraint~\eqref{eq:gibbs3},
	we have
	\begin{equation}\label{eq:gibbs5}
		\EE_\pi[|X_{\pm N}|^2 + |\tilde X_{\pm N}|^2] = 	\EE_P[|X_{-N}|^2 + |X_N|^2] + \EE_Q[|X_{-N}|^2 + |X_N|^2] \le 4E .
	\end{equation}	
	Thus, as we noticed in the proof of Theorem~\ref{th:coupling}, the constraint~\eqref{eq:gibbs5} leads to
	$$ 
\EE_\pi[\theta_c(|X_{\pm N} - \tilde X_{\pm N}|) ]
	\le f(\pi[X_{\pm N} \neq \tilde X_{\pm N}]) \leq f(\epsilon_N), $$
	where the concave non-decreasing function $f$ is
	$$ f(t) = t \theta_c\left(\sqrt{8E\over t} \right)\,. $$
	Therefore, starting from any coupling $\pi$ which achieves $\epsilon_N$ we produced a new coupling $\pi'$ which
	achieves 
		$$ \epsilon_{N-1} \leq f(\epsilon_N)\,.$$
	As we have seen in the proof of \prettyref{prop:tv-lost}, Lemma~\ref{lmm:rate} shows that such iterations lead to $\epsilon_N$ decreasing to zero. 
	Hence for any $n$, starting with
	sufficiently large $N\gg n$, we have shown that $ \TV(P_{X_{-n}^n}, Q_{X_{-n}^n}) $ is arbitrarily small, hence zero. In
	other words, distributions $P$ and $Q$ have the same finite-dimensional marginals, and must therefore coincide.
\end{proof}

As one can see our proof crucially relies on the fact that boundary of the interval $[-N,N]$ on the chain
graph~\eqref{eq:gibbs2} always consists of two points $X_{\pm N}$ (see~\eqref{eq:gibbs5}). This is why a similar argument is not applicable to Markov random fields
on $\integers^2$, where the number of variables in the boundary of $[-N,N]^2$ grows with $N$. But in that case it is well-known that even for binary-valued $X$ there can exist multiple
Gibbs measures (the two-dimensional Ising model example).

\apxonly{
Remarks about existence of $\{f_j\}$ such that Gibbs measure is not unique:
\begin{itemize}
	\item If range of all $f_j$ are all equal and have finite set as range, then there can only be one Gibbs
	measure~\cite[Chapter 3]{HOG11}.
	\item Let $f_j=0$ for $j<0$, this effectively correspond to just considering the one-directional chain $(X_0,
	Y_0)-(X_1,Y_1)-\cdots$. Also consider the following choice
	\begin{equation}\label{eq:rl6}
			f_j(Y_j) = \begin{cases}	0, &|Y_j| \le b_j\\
			a_j \cdot \sign(Y_j), & |Y_j| > b_j 
		\end{cases} 
	\end{equation}		
		where we assume $a_n\to \infty$ (so that $X_j$ has unbounded range). For convenience let us also define
		$T_j = \sign X_j$ (so $T_j = 0, \pm 1$). Then a simple argument based
		on~\eqref{eq:rl2} shows
		\begin{enumerate}
			\item $P_{T_j | T_{\hat j}} = P_{T_j | T_{j-1}, T_{j+1}}$
			\item Flip-symmetry: $P_{T_j | T_{j-1}, T_{j+1}}(t_0 | t_{-1}, t_{+1}) = P_{T_j | T_{j-1},
			T_{j+1}}(-t_0 | -t_{-1}, -t_{+1})$
			\item Zero-symmetry: $P_{T_j | T_{j-1}, T_{j+1}}(t_0 | 0,0) = P_{T_j | T_{j-1},
			T_{j+1}}(-t_0 | 0,0)$
			\item By applying technique in~\cite[Theorem 1.33]{HOG11} and since $T_{-1}=0$ 
			we can extend the last property to
			\begin{equation}\label{eq:rl7}
					P_{T_j | T_{j+N}^\infty}(t_0 | 0, t_{j+N+1}, \ldots) = P_{T_j | T_{j+N}^\infty}(-t_0 | 0,
			t_{j+N+1}, \ldots) 
		\end{equation}			
		\end{enumerate}
		Therefore, for $\{T_j, j=0,\ldots\}$ we have a positive Markov specification on $\mathbb{Z}$ (as
		in~\cite[Chapter 10-11]{HOG11}). Correspondingly, every extremal Gibbs measure must be a Markov chain
		(i.e. $P_{T_j|T_0^{j-1}} = P_{T_j|T_{j-1}}$ and $P_{T_j|T_{j+1}^\infty} = P_{T_j|T_{j+1}}$) which is
		trivial on the tail $\sigma$-algebra -- this is~\cite[Theorem 10.21]{HOG11}.

		Now, suppose that for some Gibbs measure $\PP$ we have $\PP[\liminf |T_n|=1]>0$, then since
		$a_n\to\infty$ we must have
			$$ \EE[\liminf X_n^2] = +\infty $$
		but by Fatou's lemma this violates~\eqref{eq:rl3}. Hence, every Gibbs measure satisfies
		\begin{equation}\label{eq:rl8}
				\PP[T_n = 0\mbox{-- i.o.}] = 1 
		\end{equation}		
		Let us now define another Gibbs measure $\QQ[T_0^n = t_0^n] \eqdef \PP[T_0^n = -t_0^n]$.
		From~\eqref{eq:rl7} we have that for every $t_0^n$ we have
		$$ \PP[T_0^n=t_0^n, T_{n+1}=0] = \QQ[T_0^n=t_0^n, T_{n+1}=0]\,.$$
		Finally, from~\eqref{eq:rl8} we have
		$$ \PP[T_0^n=t_0^n] = \PP[T_0^n=t_0^n, \exists k>n: T_k=0] = \QQ[T_0^n=t_0^n, \exists k>n: T_k=0]
		= \QQ[T_0^n=t_0^n] $$
		and therefore $\PP=\QQ$. \textit{Consequently, we can not argue existence of several Gibbs measures by a
		method in~\cite[Chapter 6.1]{HOG11}, since ``spin-flip'' of every Gibbs measure equals itself!}
\end{itemize}
}

\subsection{Circuits of noisy gates}\label{sec:circuits}

A circuit is a directed acyclic graph emanating from
$n$ inputs $X_1, \ldots X_n$, going through multiple intermediate nodes (``gates'') and terminating at a final node
$W$. Each gate $i$ with inputs $S_i=(S_{i,1}, \ldots, S_{i,k})$ performs a simple operation $f_i(S_i)$ and produces an output,
which is then subjected to additive Gaussian noise, so that the output value $O_i$ of the $i\Th$ gate is given by 
\begin{equation}
	O_i = f_i(S_i) + Z_i\,, \qquad Z_i \sim \matn(0,1)\,.
	\label{eq:awgn-circuit}
\end{equation}
The outputs of the $i\Th$ gate are connected to the inputs $S_j$ of subsequent gates according to the graph.
The value of $W$ is the output of the last gate.

We say that the circuit computes the Boolean function
$F:\{0,1\}^n\to\{0,1\}$ with probability of error $\epsilon$ if
$$ \PP[F(x_1,\ldots, x_n) = g(W)] \ge 1-\epsilon\,, $$
for some $g:\mreals\to\{0,1\}$ and all binary vectors $x^n$.
We assume that all gates have at most $k$ inputs. We say that the function $F$ depends essentially on input $x_i$
if there exist $x, x'\in\mreals^n$ differing in the $i\Th$ coordinate {\it only}, such that
$$ F(x) \neq F(x')\,. $$

We show below that it is not possible to have small $\epsilon$, complicated $F$, large $n$ and small power
consumed by outputs of each gate:
\begin{equation}\label{eq:nc_cost}
		\EE[|f_i(S_i)|^2] \le P\,.
\end{equation}	
This is a natural extension of the well-studied model of binary symmetric
noise (bit flips) \cite{von1956probabilistic,pippenger1988reliable}. We note that even for the settings of binary symmetric channels (BSC), quite a few open questions
remain. For example, it is known that for each $k$ there exists a threshold of maximum tolerable
noise beyond which arbitrarily complex circuits are not possible \cite{HW91,evans2003maximum}. However, this threshold is generally unknown and is 
sensitive to whether 
BSCs have crossover probability exactly $\delta$ or $\le\delta$, cf.~\cite{unger2010better}, and 
whether the output of one gate is allowed to be used at one or multiple consequent gates, cf.~\cite{evans2003maximum}. 

\begin{proposition}\label{prop:noisy} For any signal-to-noise ratio $P>0$, any Boolean function $F$ essentially depending on $n$
	inputs, and any circuits of noisy $k$-input gates computing $F$, the probability of error satisfies
	\begin{equation}
	 \epsilon \ge {1-t^*_k\over 2} + o(1)\,, \qquad n\to \infty\,, 	
	\label{eq:noisy}
\end{equation}
	where 
	$$ t^*_k = \sup\{t: \FTV(kt \wedge 1) \ge t, 0\le t\le 1\} $$
	and $\FTV(t)$ is given by~\eqref{eq:Ftv-g} with $a=P$.
\end{proposition}
For three-input gates, the lower bound \prettyref{eq:noisy} is evaluated in Fig.~\ref{fig:noisy} as a function of $P$. 

\apxonly{
\begin{remark}
Note that $t^*_k$ is the fixed point of the mapping $t\mapsto \FTV(kt\wedge 1)$ on the unit interval. Therefore $t_1^*=0$ and $t_k^*=\FTV(1)=1-2\sfQ(1/\sqrt{P})$ for sufficiently large $k$. 
	
Note that the bound of Proposition~\ref{prop:noisy} for $P>\left(Q^{-1}(1/2 - 1/(2k))\right)^2$ is equivalent
to just the lower bound on the error introduced by the last stage:
 $$ P_e \ge {1- \FTV(1, P)\over 2} =  \sfQ(1/\sqrt{P})$$
	\label{rmk:noisy}
\end{remark}
}

\begin{figure}[ht]
	\centering
	\ifpdf
	\includegraphics[width=.6\textwidth]{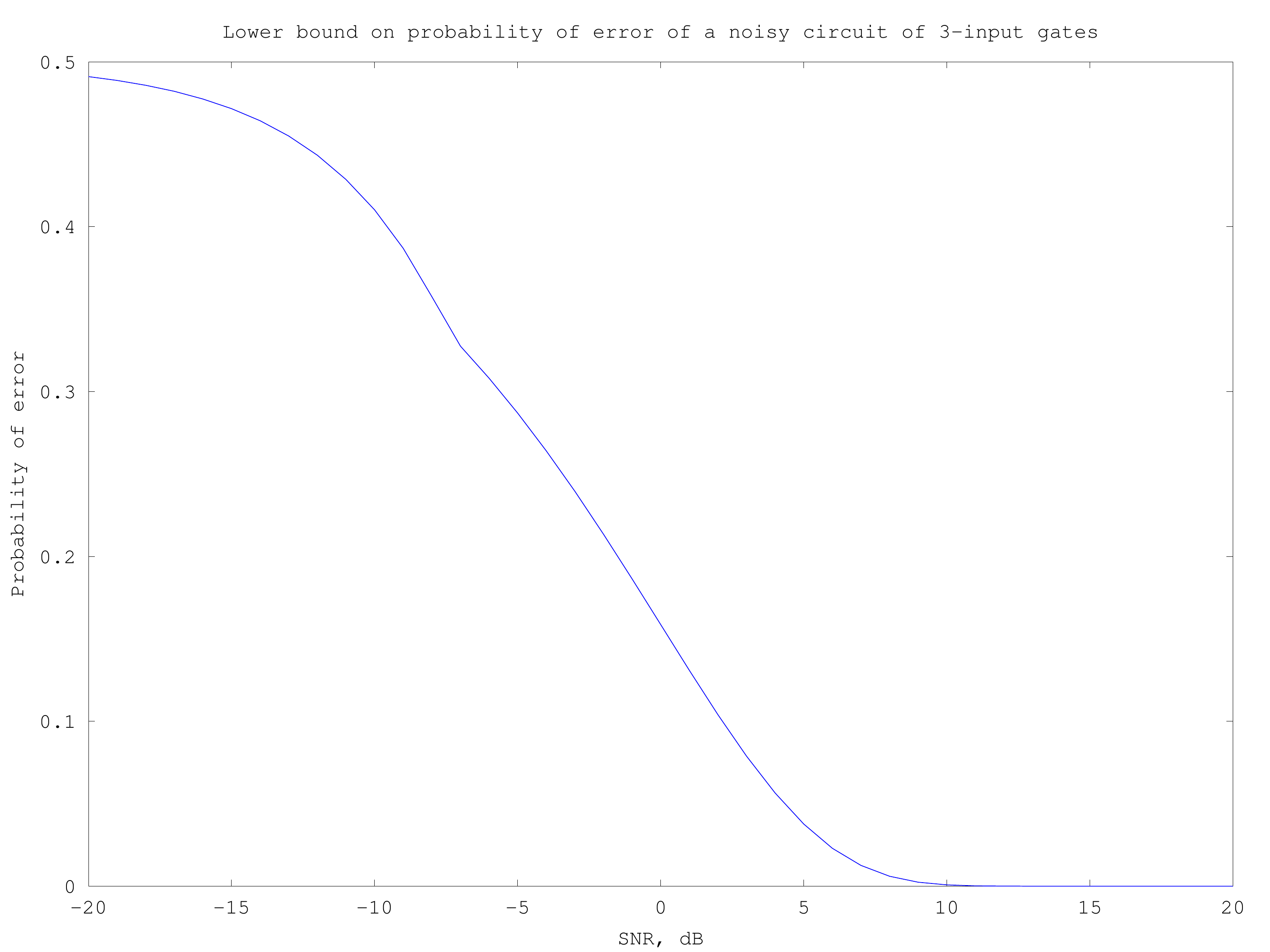}
	\else
	\includegraphics[width=.4\textwidth]{noisy_k3.eps}
	\fi
	\caption{Lower bound of probability of erroneous computation \prettyref{eq:noisy} versus signal-to-noise ratio for $k=3$.}
	\label{fig:noisy}
\end{figure}

\begin{proof} We recall a combinatorial fact shown in the proof of~\cite[Theorem 2]{evans1999signal}: 
	For every Boolean function $F$ essentially depending on $n$ inputs, and for every circuit that computes $F$ with probability of error strictly less than $\frac{1}{2}$,	there must exist at least one input, say $X_1$, such that \emph{every} path from $X_1$ to $W$ has length at least
	\begin{equation}\label{eq:pathlen}
			\ell \ge {\log n\over \log k}\,.  
	\end{equation}	
	Since $F$ essentially depends on $X_1$, we can assume, without loss of generality, that 
	$$ F(0, 0,\ldots, 0) \neq F(1, 0, \ldots, 0). $$
	
	Note that the random variables in the circuit consist of the inputs $X=(X_i)$, inputs $S=(S_i)$ and outputs $O=(O_i)$ of the gates, and the final output $W$, which is equal to some $O_i$.
	To simplify notation, let $O_0 = X_1$. Denote the neighbors of the gate $i$ by 
	$$N_i = \{j \geq 0: O_j \in S_i\},$$ whose outputs serve as inputs to gate $i$. Then $|N_i| \leq k$ by assumption.
	Without loss of generality, we assume that all gates are numbered so that $i\Th$ gate's inputs all come 
	from gates with indices strictly less than $i$. Then $N_i \subset \{0,\ldots,i-1\}$ by construction.
	
	Consider now probability distributions $P$ and $Q$ of all random variables in the circuit, such that
	under $P$ we have $X_1=0$ and under $Q$ we have $X_1=1$, while $X_2 = \cdots=X_n=0$ under both. 
	The idea is to progressively build
	coupling between $P$ and $Q$ to show that 
	\begin{equation}\label{eq:nctv1}
			\TV(P_W, Q_W) \le t^*_k + o(1)\,, 
\end{equation}	
	from which the desired lower bound \prettyref{eq:noisy} follows.

	To prove \prettyref{eq:nctv1}, suppose that there is a joint distribution  $\pi$ 
	such that 
	$$(X, O, S, W) \sim P, \quad (X', O', S', W') \sim Q,$$
	 i.e. $\pi$ is a coupling of $P$ to $Q$. 
	Consider an arbitrary gate $i$ with input $S_i$ and
	output $O_i$. In view of the noise model \prettyref{eq:awgn-circuit}, 
	the proof of Theorem~\ref{thm:gibbs} shows
	that the moment constraint~\eqref{eq:nc_cost} enables us to use \prettyref{pr:dobr}
	to build another coupling $\tilde\pi$, such that 
	a) $(X, X', O_{<i}, O'_{<i}, S_{\le i}, S_{\le i}')$ have identical joint
	distribution under either $\pi$ or $\tilde \pi$,
	and b) at the $i\Th$ gate we have 
	\begin{equation}\label{eq:nc_a1}
			\tilde\pi[O_i \neq O_i'] \le \FTV(\pi[S_i \neq S_i'])\,. 
\end{equation}
	Recall that $X_2=X_2'=\ldots=X_n=X'_n=0$ under $\pi$. Then 
	$S_i$ is determined by the outputs of the neighboring gates and possibly $O_0=X_1$, collectively denoted by $\{O_{j}: j \in N_i\}$. By the union bound, we have
	$$ \pi[S_i \neq S_i'] \le k \max_{j \in N_i} \pi[O_j \neq O'_j]\,.$$
	So if we introduce the function
	\begin{equation*}\label{eq:nc_a0}
			F_k(t) \eqdef \FTV(k t \wedge 1),
\end{equation*}	
	then we can relax~\eqref{eq:nc_a1} to
	\begin{equation}\label{eq:nc_a2}
		\tilde\pi[O_i \neq O_i'] \le F_k\Big(\max_{j \in N_i} \pi[O_j \neq O'_j]\Big).
	\end{equation}	
	
	Now, let $\pi_0$ be the trivial (independent) coupling. Since $X_1=1$ and $X_1'=0$ under $\pi_0$, we have $\pi_0[O_0 \neq O_0'] = 1 \triangleq t_0$.
	Consider the first gate, whose inputs can be either $X_1$ or
	constants. Applying the previous construction yields a coupling $\pi_1$ such that 
	$$ t_1 \triangleq \pi_1[O_1 \neq O_1'] \le F_k(1)\,. $$
	Here $t_1$ measures the quality of coupling at the output of the first gate. 
	Next, suppose that all gates $<i$ are similarly coupled by $\pi_{i-1}$ with respective $t_1,\ldots,
	t_{i-1}$. We refine the coupling at gate $i$ to get $\pi_i$, so that 
	a) the joint distribution of $(O_{<i}, O'_{<i})$ and hence $t_1,\ldots,
	t_{i-1}$ are unchanged,
	and b)
	\begin{equation}\label{eq:nc_a3}
			t_i \triangleq \pi_i[O_i \neq O_i'] \le F_k\Big(\max_{j \in N_i} t_j\Big)\,,
	\end{equation}	
	which follows from \prettyref{eq:nc_a2}.
	Continuing similarly, we arrive at the last gate which outputs $W$. Now let us construct a path from $W$ back
	to $X_1=O_0$ as follows: starting from $W$ go back from gate $i$ to the neighboring gate $j<i$ that achieves $\max_{j \in N_i} t_j$. 
	Let $m$ be the length of this path and let the indices 
	(in increasing order) be 
		$$ i_0=0< i_1=1 < i_2 < \ldots <i_m .$$ 
		By~\eqref{eq:pathlen} we must have $m = \Omega(\log n)$. 
		By construction of the path, we have $t_0=1,t_1 \leq F_k(1)$, $t_{i_2}
		\le F_k(t_1)$, etc. So finally
		$$ \pi[W\neq W'] \le t_{i_m} \leq F_k\Big(\max_{j \in N_i} t_j\Big) = F_k(t_{i_{m-1}}) \leq \ldots \le \underbrace{F_k \circ F_k \cdots \circ F_k}_{m\text{~times}}(1).$$
	Hence as $n\diverge$ this repeated composition of $F_k$'s must converge to a fixed point $t_k^*$, thus
	proving~\eqref{eq:nctv1}. 
\end{proof}

\subsection{Broadcasting on trees}\label{sec:trees}

Consider the setting studied in~\cite{EKPS00}: the original bit $W=\pm1$ is to be broadcasted along the binary
tree of noisy channels:

$$ \def\vupdown{\vbox to 0pt{\vskip -14pt\hbox{$\upto$}\hbox{$\downto$}}}
	\begin{array}{ll}
	& Y_{2,2} \to X_{2,2} \vupdown \quad \cdots \\
				W \to X_{1,1} \vupdown &  \\
				& Y_{2,1} \to X_{2,1} \vupdown \quad \cdots
			\end{array}
	$$
where arrows $X\to Y$ represent independent noisy channels and $Y\to X$ are relays. The goal is to design the 
relay functions so that for some $\epsilon>0$ one can 
reconstruct $W$ with probability of error at most ${1\over2}-\epsilon$ based on the values at the $n\Th$ layer $\{X_{n,1},\ldots,X_{n,{2^{n-1}}}\}$ for all sufficiently large $n$; to wit, the total variation of the distributions conditioned on $W=1$ or $-1$ is strictly bounded away from zero.
One of the main results of~\cite{EKPS00} is that when all channels are BSC with flip probability $\delta$ such broadcasting is possible if
and only if $2(1-2\delta)^2 > 1$, thus establishing a certain ``phase transition'' in this problem.

In fact, the impossibility part of the BSC result follows from a result of Evans and Schulman~\cite{evans1999signal}:
for a binary tree of discrete channels the probability of error tends to $1\over2$ as the depth tends to infinity
whenever $2\etaKL < 1$. For Gaussian channels we know that $\etaKL=1$ which suggests that such transition \textit{does
not occur} for a tree of Gaussian channels. Indeed, in this section we demonstrate that  it is possible to broadcast
some information to arbitrarily deep layers regardless of how small the SNR is.

Specifically, consider channels
$$ Y_{k,j} = X_{k-1, j} + Z_{k, j}\,, \qquad Z_{k,j} \iiddistr \matn(0,1) $$
with cost constraint 
\begin{equation}\label{eq:bt0a}
		\EE[|X_{k,j}|^2] \le E\,. \qquad \forall k,j
\end{equation}
Choose the initial (randomized) encoder as follows:
$$ X_{1,1} = \mu B W, \qquad \PP[B=1] = 1-\PP[B=0] = 2p, B\dperp W, $$
with parameters $p, \mu$ to be specified later. Similar to the scheme in \prettyref{sec:ach}, choose relays as follows:
$$ X_{k,j} = \begin{cases} +\mu, & Y_{k,j} \ge t\mu\,,\\
	0, &|Y_{k,j}| < t\mu\,,\\
	-\mu, & Y_{k,j} \le -t\mu\,,
\end{cases}  $$
where $t \in (1/2, 1)$ can be set arbitrarily. Notice that if $\mu$ is selected so that
\begin{equation}\label{eq:bt0}
		p = {Q(t \mu) \over Q((1-t) \mu) + 2Q(t\mu) - Q((1+t)\mu) } 
\end{equation}
then a simple computation shows that for all $k,j$ we have
\begin{equation}
	\PP[X_{k,j} = +\mu] = \PP[X_{k, j}=-\mu] = p\,.
	\label{eq:xkj}
\end{equation}  
But from~\eqref{eq:bt0} and the fact that $t>{1/2}$ for large $\mu$ we get
$$ p = e^{-\mu^2 (t-1/2) +O(1) }\,, \qquad \mu \to \infty .$$
In particular, regardless of how small $E$ in~\eqref{eq:bt0a} is and for any $t$, there exists a sufficiently large $\mu$ 
such that the cost constraint is satisfied. Another important parameter turns out to be
$$ \theta = 1-Q((1-t) \mu) - Q((1+t) \mu).$$
Again, taking $\mu$ large we may ensure
\begin{equation}\label{eq:bt0b}
		2\theta^2 > 1. 
\end{equation}	
Thus we assume from now on that $p,\mu$ and $t$ are selected in such a way that both~\eqref{eq:bt0a} and~\eqref{eq:bt0b} are
satisfied.

Similarly to~\cite{EKPS00} we will employ the idea of T. Kamae, see~\cite[Remark on p. 342]{YH77}, and consider the behavior
of ``spin sums'':
$$ S_k = \sum_{j=1}^{2^{k-1}} \sigma_{k, j}\,, $$
where $\sigma_{k, j} \triangleq \sign(X_{k, j})$ with $\sign(0)=0$, or equivalently, $\sigma_{k, j}=X_{k, j}/\mu$. To show that it is possible to test $W=\pm 1$ based on the statistic $S_n$, we show that
\begin{equation}\label{eq:bt1}
	\liminf_{n\to \infty} \TV(P_{S_n|W=+1}\| P_{S_n|W=-1}) \ge 2p  \left(1 - {1\over 4\theta^2}\right)\,,
\end{equation}
which is strictly positive. According to~\cite[Lemma 4.2 (i) and (iii)]{EKPS00} we have:
$$ \TV(P_{S_n|W=+1}\| P_{S_n|W=-1}) \ge {(\EE[S_n|W=+1] - \EE[S_n|W=-1])^2\over 4 \EE[S_n^2]}. $$
So the estimate~\eqref{eq:bt1} follows from two results:
\begin{align} 
\EE[S_n | W= \pm 1]  &= \pm 2p(2\theta)^{n-1} \label{eq:Sn1} ,\\
\EE[S_n^2] &\le 2^{n} p + 2 p \frac{(2\theta)^{2n}}{(2\theta)^2 - 1} \label{eq:Sn2}.
\end{align}
Both of these are verified below: 
Consider two arbitrary nodes $(k,j)$ and
$(k,j')$ at the $k\Th$ level and let $(u,i)$ be their common ancestor in the tree. Denote the parent node of $(k,j)$ by $(k-1,j'')$.
Then
$$ \EE[\sigma_{k,j} | \sigma_{u,j'}] = \Expect[\EE[\sigma_{k,j} | \sigma_{k-1,j''}] | \sigma_{u,j'}] =  \theta \, \EE[\sigma_{k-1,j''} | \sigma_{u,j'}] = \ldots =
 \theta^{k-u} \sigma_{u,i}. $$
 Furthermore, $\sigma_{k,j}$ and $\sigma_{k,j'}$ are independent conditioned on $\sigma_{u,i}$.
Note that $\Expect[\sigma_{1,1}|W=\pm 1] = \pm \Prob[B=1] = \pm 2p$, which yields \prettyref{eq:Sn1}. Next, note that
$\Expect[S_n^2] = \sum_{j=1}^{2^{k-1}} \Expect[\sigma^2_{n, j}] + 2 \sum_{j'<j} \Expect[\sigma_{n, j} \sigma_{n, j'}]$, where the first term is $2^{n-1}\times 2p$ since $\sigma^2_{k,j} \sim \Bern(2p)$ in view of \prettyref{eq:xkj}. To estimate the cross term, denote the depth of the common ancestor of $(n,j)$ and $(n,j')$ by $u(j,j') \in \{1,\ldots, n-1\}$. Then
\begin{align}
\sum_{j'<j} \Expect[\sigma_{n, j} \sigma_{n, j'}]
= & ~ \sum_{u=1}^{n-1} \sum_{u(j',j)=u} \Expect[\sigma_{n, j} \sigma_{n, j'}]	 = \sum_{u=1}^{n-1} \sum_{u(j',j)=u} \theta^{2(n-u)} 2 p	\\
= & ~ 2 p \sum_{u=1}^{n-1} \theta^{2(n-u)} \binom{2^{n-u}}{2}  \leq p \frac{(2\theta)^{2n}}{(2\theta)^2-1},
\end{align}
which yields \prettyref{eq:Sn2}.


\section*{Acknowledgment}
It is a pleasure to thank Max Raginsky (UIUC) for many helpful discussions and Flavio du Pin Calmon (MIT) for 
Proposition~\ref{prop:ITV}.

\appendices

\section{Convergence rate analysis}
	\label{app:rate}
	
	Consider the following iteration
\[
t_{n+1} = t_{n} - h(t_n), \quad t_1=1
\]
where $h:[0,1] \to [0,1]$ satisfies $h(0)=0$ and $0 < h(t) \leq t$ for all $0<t\leq 1$. 
Then $\{t_n\}\subset[0,1]$ a monotonically decreasing sequence converging to the unique fixed point zero as $n\diverge$. 
Under the monotonicity assumption of the function $h$, the following result gives a non-asymptotic upper estimate of this sequence.
	\begin{lemma}
Define $G: [0,1] \to \reals_+$ by $G(t) = \int_t^1 \frac{1}{h(\tau)} \diff \tau$. If $h$ is increasing, then for any $n\in \naturals$,
\begin{equation}
	t_n \leq G^{-1}(n-1).
	\label{eq:tn-rate}
\end{equation}	
	\label{lmm:rate}
\end{lemma}
\begin{proof}
By the positivity and monotonicity of $h$, $G$ is a strictly decreasing and concave function. Hence $G^{-1}: \reals_+ \to [0,1]$ is well-defined.
Put $b_n=G(t_n)$. Then 
	\begin{align}
	b_n - b_{n-1} = \int_{t_n}^{t_{n-1}} \frac{1}{h(\tau)} \diff \tau \geq \frac{t_{n-1} - t_n}{h(t_{n-1})} = 1.
\end{align}
Hence $b_n\geq n-1$ since $b_1 = G(1)=0$.
\end{proof}

\section{Contraction coefficient for mutual information: General case}
	\label{app:etaKL}
	We shall assume that $P_X$ is not a point mass, namely, there exists a measurable set $E$ such that $P_X(E) \in (0,1)$. 
	Define
\[
\etaKL(P_X) = \sup_{Q_X} \frac{D(Q_{Y} \| P_{Y})}{D(Q_X\|P_X)}
\]
where the supremum is over all $Q_X$ such that $0<D(Q_X\|P_X) < \infty$. It is clear that such $Q_X$ always exists (\eg, $Q_X = P_{X|X\in E}$ and $D(Q_X\|P_X) = \log \frac{1}{P_X(E)} \in (0,\infty)$). Let
\[
\eta_I(P_X) = 	\sup \frac{I(U; Y)}{I(U;X)}
\]
where the supremum is over all Markov chains $U\to X \to Y$ with fixed $P_{XY}$ such that $0<I(U;X)<\infty$. Such Markov chains always exist, \eg, $U=\indc{X\in E}$ and then $I(U;X)=h(P_X(E)) \in (0,\log 2)$.
The goal of this appendix is to prove~\eqref{eq:etaKL}, namely
$$ \etaKL(P_X) = \eta_I(P_X)\,.$$

The inequality $\eta_I(P_X) \leq \etaKL(P_X)$ follows trivially:
$$ I(U;Y) = D(P_{Y|U} \| P_Y |P_U) \le \etaKL(P_X) D(P_{X|U}\|P_X | P_U) = \etaKL(P_X) I(X;U)\,.$$

For the other direction, fix $Q_X$ such that $0<D(Q_X\|P_X) < \infty$. First, consider the case where $\fracd{Q_X}{P_X}$
is bounded, namely, $\fracd{Q_X}{P_X} \leq a$  for some $a>0$ $Q_X$-a.s. 
For any $\epsilon \leq \frac{1}{2a}$, let $U \sim \Bern(\epsilon)$ and define the probability measure $\tilde P_X = \frac{P_X-\epsilon Q_X}{1-\epsilon}$. Let $P_{X|U=0} = \tilde P_X$ and $P_{X|U=1} = Q_X$, which defines a Markov chain $U \to X \to Y$ such that $X,Y$ is distributed as the desired $P_{XY}$. Note that 
\[
\frac{I(U; Y)}{I(U;X)} = \frac{\bar \epsilon D(\tilde P_Y\|P_Y) + \epsilon D(Q_Y\|P_Y)}{\bar \epsilon D(\tilde P_X\|P_X) + \epsilon D(Q_X\|P_X)}
\]
where $\tilde P_Y  = P_{Y|X} \circ \tilde P_X$.
We claim that 
\begin{equation}
	D(\tilde P_X\|P_X) = o(\epsilon),
	\label{eq:oepsilon}
\end{equation} which, in view of the data processing inequality $D(\tilde P_X\|P_X) \leq D(\tilde P_Y\|P_Y)$, implies $\frac{I(U; Y)}{I(U;X)} \xrightarrow{\epsilon \downarrow 0} \frac{D(Q_Y\|P_Y)}{D(Q_X\|P_X)}$ as desired. 
To establish \prettyref{eq:oepsilon}, define the function
$$ f(x,\epsilon) \eqdef \begin{cases} {1-\epsilon x\over \epsilon(1-\epsilon)} \log{1-\epsilon x\over 1-\epsilon}\,, &
\epsilon > 0 \\
	(x-1) \log e, &\epsilon=0\,. \end{cases} $$
One easily notices that $f$ is continuous on $[0,a]\times[0,{1\over 2a}]$ and thus bounded. So we get, by
bounded convergence theorem,
$$ \frac{1}{\epsilon} D(\tilde P_X\|P_X) = \EE_{P_X}\left[f\left(\fracd{Q_X}{P_X}, \epsilon\right) \right] \to  \EE_{P_X}\left[\fracd{Q_X}{P_X}-1\right] \log e = 0\,. $$

To drop the boundedness assumption on $\fracd{Q_X}{P_X}$ we simply consider the conditional distribution $ Q_X' \eqdef Q_{X|X\in A}$
where $A=\{x: \fracd{Q_X}{P_X}(x) < a\}$ and $a>0$ is sufficiently large so that $Q_X(A) > 0$. Clearly, as $a\diverge$, we have $Q_X' \to Q_X$ and $Q'_Y \to Q_Y$
pointwise, where $Q'_Y \triangleq P_{Y|X} \circ Q'_X$. Hence the lower-semicontinuity of divergence yields
$$ \liminf_{a\to\infty} D(Q_Y'\|P_Y) \ge D(Q_Y\|P_Y)\,.$$
Furthermore, since $\fracd{Q'_X}{P_X} = \frac{1}{Q_X(A)} \fracd{Q_X}{P_X} \Indc_{A}$, we have
\begin{align}
D(Q'_X\|P_X)
= & ~ \log \frac{1}{Q_X(A)} + \frac{1}{Q_X(A)} \EE_Q\left[\log \fracd{Q_X}{P_X} \mathbf{1}\left\{\fracd{Q_X}{P_X} \leq a\right\}\right].
\end{align}
Since $Q_X(A)\to 1$, by dominated convergence (note: $\EE_Q[|\log \fracd{Q_X}{P_X}|] < \infty$) we have
$ D(Q'_X\|P_X) \to D(Q_X\|P_X) $.
Therefore,
$$ \liminf_{a\to\infty} {D(Q_Y'\|P_Y)\over D(Q_X'\|P_X)} \ge {D(Q_Y\|P_Y)\over D(Q_X\|P_X)}\,,$$
completing the proof.

\newcommand{\etalchar}[1]{$^{#1}$}

\ifmapx
	\pagebreak
	\section{\Large\bf APXONLY STUFF}
	\subsection{How far can we send stuff?}

	Goal: Build a chain of relays to send 1 bit from our solar system to another one. 
	
	Model: Each
	receiver has noise level $N_0 \approx kT$ (with $T=4K$ for background radiation).
	Transmitters have powers $P_k$ and are spaced at distances $r_k$ (so that $SNR_j =
	\mathrm{const}\cdot{P_k\over r_k^2 N_0}$).

	We have shown previously that taking $P_k=r_k=1$ does not lead to reliable
	communication. I.e. $\TV(P_\infty, Q_\infty)=0$.

	Question: Is it possible that we have simultaneously:
	\begin{align}
		\TV(P_\infty, Q_\infty) &> 1/2\label{eq:gal1}\\
		\sum_{k=1}^\infty r_k &= \infty\label{eq:gal2}\\
		\sum_{k=1}^\infty P_k &< \infty 
	\end{align}	
	In words: is it possible to pass a bit infinitely far with finite
	total-energy-per-bit?

	Answer: Yes. Indeed, according to~\eqref{eq:Ftv} at each iteration the total variation contracts by at most
		$$ 1-2Q({\sqrt{P_k}\over r_k}) $$
		(I assume there is a scheme to achieve this).
	Set
		$$ P_k = {1\over k \ln k}, \quad r_k = {c\over \sqrt{k} \ln k}\,. $$
	Then we have for all $c < {1\over2}$ 
		$$ \prod_{k=1}^\infty (1-2Q({\sqrt{P_k}\over r_k}))  > 0 $$
	since
	$$ \sum_{k=1}^\infty Q({\sqrt{P_k}\over t r_k})) \sim 
	\sum_{k=1}^\infty e^{-{1\over 2c} \ln k} < \infty $$
	whenever ${1\over 2c}>1$. 

	Furthermore, evidently the $\sum_{k=1}^\infty P_k$ can be made arbitrarily small,
	implying a bit can be transferred to infinity with zero energy expended, provided
	there are relays along the path. \textbf{TODO:} Doesn't it contradict minimum
	energy-per-bit concept?

\subsection{Open questions}
	\begin{enumerate}
		\item Strong conjecture: there is a positive function
	$f:(0,1)\to(0,1)$ such that for all $(X,Y)$
		$$ I(X;Y) \ge \epsilon_0 \implies \sup_{E} I(1_E(X); Y) \ge
		f(\epsilon_0)>0 $$
	%
	%
	This conjecture is \textbf{FALSE}. Consider, $X=Y$ and $P_X=[1-\delta,
	{\delta\over k}, \ldots, {\delta\over k}]$ with $k=\delta 2^{1\over \delta}$. Then
		$$ I(X;Y) = H(X) = \log 2 - (1-\delta) \log(1-\delta) \ge
		1\mbox{~bit}\,.$$
		On the other hand, it is easy to show that
		$$ \sup_{B\to X\to Y, B\in\{0,1\}} I(B; Y) = \sup_{q:\matx\to\{0,1\}} I(q(X); Y) $$
		On the other hand, for the current example
		$$ I(q(X); Y) = H(q(X)) \le h(\delta)\,.$$

	\item Find
		$$ F_I(t) = \sup_{U-X-Y: I(U;X)\le t} I(U;Y) $$
		For binary $U$ the question boils down to contraction of some funny $f$-divergence and $F_I(t)<t$
		follows from integral representation via $\mate_\gamma$'s. For general, not clear what to do.

		\item Another problem: If we are interconnected by AWGNs of infinite dimension (and a bound on total
		energy) then we \textit{do not know} whether $I(W; Y_n)\to0$ or not. This is because our upper
		bound~\eqref{eq:main3} is of the form $d^2 \cdot E$ as opposed to $d\cdot E$. So our $\TV\to I$
		technique should be changed for infinite dimension. Is this related to Rakhlin's fetish with
		Kolmogorov-packings of infinite dimensional spaces shit?

		Later: Also note that $d^2 \cdot E$ corresponds to dimension of $X_1$, not subsequent $X_j$'s.

		\item Our result on circuits suggests there is no phase transition for AWGNs (like for BSC: $\delta>1/6$
		implies $P_e \to 1/2$ as $n\to\infty$). 
		So we need to make von~Neumann's construction for AWGN noise and show that $P_e <1/2$ no matter
		how low SNR is.
	\end{enumerate}

\subsection{Dobrushin and hyper-contractivity}
	Here is a funny argument I noticed. Let $P_Z$ be noise on some abelian group and suppose that 
	$$ \|P_Z - U_Z\|_{TV} < 1/2 $$
	where $U_Z$ is the Haar measure on the group. Then the channel
		$$ X \to X+Z$$
	is TV-contractive since 
		$$ \etaTV = \max_v \|P_{Z+v} - P_Z\|_{TV} \le 2\|P_Z - U_Z\|_{TV} < 1 $$
	Thus by~\cite{CKZ98} we also have
		$ \etaKL < 1 $ and by Ahslwede-G\'acs the operator
		$$ Tf(x) \eqdef f*P_Z $$
	is hypercontractive $L_p(P_Y) \to L_q(P_X)$ for EVERY $P_X$! This may be quite tough to show in general,
	while condition $\|P_Z-U_Z\|_{TV}<1/2$ is really nice. 

	\textbf{Question:} Is there some hypercontractivity implication we may extract from Dobrushin curve?

\fi
	
\end{document}